\documentclass{IEEEoj}
\usepackage{cite}
\usepackage{amsmath,amssymb,amsfonts}
\usepackage{algorithmic}
\usepackage{graphicx}
\usepackage{textcomp}
\usepackage{amsthm}
\usepackage{algorithm}
\usepackage{xcolor}
\usepackage{latexsym}
\usepackage{float}
\usepackage{indentfirst}
\usepackage{mdwmath}
\usepackage{mdwtab}
\usepackage{subfig}
\usepackage{setspace}
\usepackage{stfloats}
\usepackage{url}
\usepackage{multirow}

\usepackage[colorlinks=true, allcolors=blue]{hyperref}

\def\BibTeX{{\rm B\kern-.05em{\sc i\kern-.025em b}\kern-.08em
    T\kern-.1667em\lower.7ex\hbox{E}\kern-.125emX}}
\AtBeginDocument{\definecolor{ojcolor}{cmyk}{0.93,0.59,0.15,0.02}}
\def\OJlogo{\vspace{-4pt}\hskip-4pt\includegraphics[height=18pt]{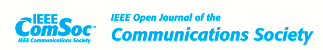}}

\newtheorem{remark}{Remark}
\newtheorem{theorem}{Theorem}

\newtheorem{lemma}{Lemma}

\newtheorem{corollary}{Corollary}

\newtheorem{proposition}{Proposition}

\begin{document}
\receiveddate{XX Month, XXXX}
\reviseddate{XX Month, XXXX}
\accepteddate{XX Month, XXXX}
\publisheddate{XX Month, XXXX}
\currentdate{16 March, 2026}
\doiinfo{OJCOMS.2026.011100}

\title{On the Performance of Single/Dual Fluid Antenna Systems}

\author{JIANGSHENG HUANGFU\IEEEauthorrefmark{1}, ZHENGYU SONG \IEEEauthorrefmark{1} \IEEEmembership{(Member, IEEE)}, \\
TIANWEI HOU \IEEEauthorrefmark{1} \IEEEmembership{(Member, IEEE)}, ANNA LI \IEEEauthorrefmark{2} \IEEEmembership{(Member, IEEE)}, YUANWEI LIU \IEEEauthorrefmark{3} \IEEEmembership{(Fellow, IEEE)}, \\
AND ARUMUGAM NALLANATHAN\IEEEauthorrefmark{4,5}
\IEEEmembership{(Fellow, IEEE)}}
\affil{School of Electronic and Information Engineering, Beijing Jiaotong University, Beijing 100044, China}
\affil{School of Computing and Communications, Lancaster University, Lancaster LA1 4WA, U.K.}
\affil{Department of Electrical and Electronic Engineering, The University of Hong Kong, Hong Kong}
\affil{School of Electronic Engineering and Computer Science, Queen Mary University of London, London, U.K.}
\affil{Department of Electronic Engineering, Kyung Hee University, Yongin-si, Gyeonggi-do 17104, Korea}
\corresp{CORRESPONDING AUTHOR: TIANWEI HOU (e-mail: twhou@bjtu.edu.cn).}
\authornote{This work was supported in part by the Fundamental Research Funds for the Central Universities under Grant 2024JBMC014 and 2023JBZY012, in part by the Beijing Natural Science Foundation L232041, in part by EPSRC grant numbers to acknowledge are EP/W004100/1, EP/W034786/1, EP/Y037243/1 and EP/W026813/1. }
\markboth{Preparation of Papers for IEEE OPEN JOURNALS}{Author \textit{et al.}}

\begin{abstract}
The emerging technology of fluid antenna systems (FASs) represents a promising next-generation reconfigurable antenna solution, capable of exploiting the full spatial diversity within a predefined space by finely reconfiguring the positions of radiating elements. In this paper, the performance of FAS over spatially correlated Rayleigh fading channels is investigated for two distinct scenarios: a multiple-input single-output (MISO) configuration, where a receiver with a single-antenna FAS is served by a multi-antenna transmitter (MISO-FAS), and a single-input single-output setup where single-antenna FASs are equipped at both the transmitter and receiver (Dual-FAS).
Exact expressions and closed-form approximations for the outage probability (OP) of both the MISO-FAS and Dual-FAS models are derived as the core contributions of this work. To provide deeper insights into system performance, the diversity orders for each model are also derived and analyzed. Analytical results demonstrate that increasing the number of ports significantly enhances system performance.
The theoretical analysis is corroborated by key findings from our simulations, demonstrating that:
$i$) Both the MISO-FAS and Dual-FAS models achieve considerable performance gains as the number of ports is increased;
$ii$) System performance for both configurations is inversely related to the level of port correlation; lower correlation leads to better performance;
$iii$) In the high signal-to-noise ratio regime, the Dual-FAS model surpasses the performance of the MISO-FAS model.
\end{abstract}


\maketitle

\section{INTRODUCTION}
Multiple-input multiple-output (MIMO) has long been a key technology in wireless communication. Its main advantage is the ability to exploit the spatial domain for transmitting information, in addition to the traditional time and frequency domains~\cite{MIMO0,MIMO1,MIMO2}. 
In current fifth-generation (5G) networks, massive MIMO has emerged as a cornerstone technology and has been widely studied and deployed~\cite{Marzetta2010, Massive1}. The base station (BS) employs a very large array of antennas, which significantly improves spectral efficiency, network capacity, and reliability~\cite{Massive2}. Massive MIMO operates by focusing energy directly towards users, which increases data rates and enhances overall performance~\cite{Massive3}. 
As the natural evolution of 5G, sixth-generation (6G) networks are expected to bring revolutionary advances, which critically rely on the increase in degrees of freedom (DoF) in the physical layer~\cite{6G1,6G2,6G3}. Among various emerging technologies, extra-large scale MIMO (XL-MIMO) has been widely regarded as a promising enabler~\cite{Wang-xlmimo,Wang2024}.
Other technologies including non-orthogonal multiple access (NOMA) \cite{NOMA} and reconfigurable intelligent surfaces (RIS) \cite{RIS} are receiving significant attention. 
Combining these technologies with the MIMO framework is a particularly promising strategy. Recent studies demonstrate that this integration can further reduce interference~\cite{MIMO_RIS_NOMA,MIMO_LIS} and substantially enhance signal quality~\cite{MIMO_RIS}.

In contrast to the BS, which can support massive antenna arrays, user equipment (UE) is fundamentally limited by its compact physical dimensions, particularly in devices such as mobile phones, wearable electronics, and Internet of Things (IoT) sensors~\cite{Massive1}. This severe spatial constraint directly restricts the number of fixed antennas that can be accommodated on a single device, creating a primary bottleneck that limits the end-to-end communication performance. To ensure independent channel fading and achieve spatial diversity, conventional fixed antennas typically require a minimum spatial separation, often on the order of half a wavelength~\cite{Sharawi2013}. Consequently, simply increasing the antenna density within a confined UE area is not a viable solution, as it inevitably introduces strong mutual coupling between adjacent antenna elements. This mutual coupling alters radiation patterns, degrades antenna efficiency, and increases spatial correlation, thereby further degrading the overall system performance rather than improving it~\cite{Sharawi2013,couple}. Furthermore, the feasibility of integrating a large number of conventional antennas at the UE is heavily constrained by hardware limitations. Each active antenna typically requires a dedicated radio-frequency (RF) chain. Scaling up the number of RF chains incurs prohibitive hardware costs, and imposes unacceptable strains on the stringent power budgets of portable and battery-powered IoT devices~\cite{Alkhateeb2014, Han2015,expensive}.

These limitations have motivated the search for alternative antenna technologies capable of delivering substantial performance gains within a compact physical space. In this context, the fluid antenna system (FAS) has emerged as a promising and flexible paradigm for future 6G networks \cite{FAS}. 
Fundamentally, a FAS is a software-controllable radiating structure that can change its position within a confined region. The core principle is analogous to classical spatial diversity, but with a revolutionary difference: instead of using multiple fixed antennas to sample the channel at different points, a FAS employs a single radiating element that can move to any of numerous candidate locations, known as ports \cite{FAS1,FAS2}. This unique capability allows the system to dynamically identify and select the port with the most favorable channel conditions, thereby enabling a more complete exploitation of spatial diversity \cite{block}.
Enabled by recent advances in reconfigurable antenna technologies, the physical realization of FAS can be achieved through various mechanisms.
One fundamental approach to realize the FAS involves using mechanically movable antennas. In this setup, stepper motors drive RF elements to reconfigure antenna positions in space~\cite{mecha1,mecha2,mecha3}. An alternative approach utilizes the inherent flexibility of liquid materials as radiating antenna elements. Both metallic and non-metallic liquids can be employed in the design of FAS~\cite{liquid1,liquid2,liquid3}. More recently, liquid-based FAS designs have been proposed, where micro-pumps are employed to control the movement of the radiating elements within pre-designed fluidic channels \cite{liquidFAS1,liquidFAS2,liquidFAS3}. A pump-free electro-wetting-on-dielectric (EWOD) technique was adopted in \cite{EWOD}, achieving a motion speed of 10 mm/s. Another prominent implementation is based on pixel reconfigurable antennas (PRAs), where electronic switches are used to reconfigure the antenna structure, enabling a response time on the order of milliseconds or less \cite{pixel1,pixel2,pixel3}. This feature allows PRAs to rapidly adapt to variations in channel state information (CSI), thereby fully exploiting the potential of the FAS \cite{tutorial}.

\subsection{MOTIVATIONS AND CONTRIBUTIONS}
FAS was first proposed as a novel concept for wireless communications in 2020~\cite{FAS, FAS_ER}. Following its inception, several tutorial letters and comprehensive surveys have summarized the fundamentals, channel models, and open research opportunities of FAS technology~\cite{tutorial, FAS_partI, FAS_partII}. A critical foundational topic in FAS research is the accurate modeling of spatial correlation among the densely packed ports, which has been extensively investigated in recent literature~\cite{block, wong2022closed, Khammassi-2023}. 
Depending on the architecture, a FAS can be deployed at the transmitter, the receiver, or both. In single-user scenarios, the majority of research has focused on the Rx-SISO-FAS configuration, where the receiver is equipped with a FAS and the transmitter uses a single fixed antenna. Early works established the performance baseline by analyzing FAS over Rayleigh fading channels to derive the outage probability (OP)~\cite{FAS} and a lower bound for the ergodic rate (ER)~\cite{FAS_ER}. More recently, the diversity order of FAS was rigorously established, and exact closed-form expressions for the OP, ER, and diversity gains were derived, providing accurate and tractable performance characterizations~\cite{Zhao2025FAS, FAS_OP_diversity}.
Subsequent studies have extended these analyses to more general fading environments. For instance, the OP analysis was expanded to Nakagami-$m$ fading channels in~\cite{nakagami1}, with novel asymptotic matching methods~\cite{nakagami2} and Gaussian copula functions~\cite{copula} employed to handle correlated Nakagami-$m$ fading. Furthermore, performance under Rician fading was investigated in~\cite{rician}, concluding that a strong line-of-sight (LoS) component tends to degrade the FAS performance. The impact of nonlinear fading was also explored via $\alpha$-$\mu$ channels in~\cite{FAS_alpha_mu}. Beyond outage and ergodic-rate metrics, the average symbol error rate (ASER) of FAS utilizing rectangular and hexagonal quadrature amplitude modulation (QAM) schemes was evaluated in~\cite{FAS_ASER}.
While the aforementioned studies primarily consider scenarios where only a single port of the FAS is activated, recent works have explored multi-port activation, such as combining received signals via maximal-ratio combining (MRC)~\cite{MRC}. However, for more complex setups, theoretical investigations remain limited. For example, the performance limits of MIMO-FAS systems—where FASs are deployed at both the transmitter and receiver—have only recently been analyzed using optimization techniques~\cite{trade}.

Motivated by the aforementioned challenges, in this paper, we provide a rigorous comparative analysis of two distinct FAS-enabled communication models:
\begin{enumerate}
    \item MISO-FAS: A multiple-input single-output (MISO) configuration where a transmitter equipped with multiple fixed-position antennas communicates with a receiver equipped with FAS.
    \item Dual-FAS: A single-input single-output (SISO) configuration where both the transmitter and the receiver are equipped with a FAS.
\end{enumerate}

The main contributions of this paper are summarized as follows:

\begin{itemize}

\item We consider a system where the transmitter is equipped with $N$ independent fixed-position antennas employing maximum ratio transmission (MRT) precoding, and the receiver is equipped with a FAS that can switch among $M$ ports to obtain the maximum gain. By accounting for spatial correlation, we derive the joint probability density function (PDF) and cumulative distribution function (CDF) of the channel gain.

\item We further consider the case where both the transmitter and receiver are equipped with FASs, with $M_T$ and $M_R$ ports, respectively. Taking channel correlation into account, we derive the joint PDF and CDF of the channel gain.

\item We derive the exact expressions for the OP of both models from the obtained PDF and CDF. In addition, closed-form expressions for the upper and lower bounds are derived. From the upper bound, the diversity orders are further characterized, which are determined by the sum of the number of ports and antennas for MISO-FAS, and by the product of the number of the transmitter and receiver ports for Dual-FAS.
    
\item Our simulation results demonstrate that 1) Increasing the number of ports significantly enhances system performance; 2) Stronger port correlation degrades system performance; 3) At low SNR, MISO-FAS achieves better performance, while Dual-FAS exhibits faster performance improvement at high SNR; 4) With the same number of receiver-side FAS ports, Dual-FAS outperforms MISO-FAS by increasing the number of transmitter-side ports, despite utilizing fewer RF chains.

\end{itemize}

\subsection{ORGANIZATIONS AND NOTATIONS}
Section \ref{sec:model} introduces the channel models for MISO-FAS and Dual-FAS. Based on these models, Section \ref{sec:analysis} presents the derivation of the analytical results. Subsequently, these results are validated through numerical simulations in Section \ref{sec:simu}. Finally, Section \ref{sec:conclude} summarizes the key findings and concludes the paper.

Throughout this paper, the following notation is adopted: \( \mathbb{E}(\cdot) \) denotes the expectation. \( |\cdot| \) denotes the modulus. \( X|Y \) signifies \( X \) conditioned on \( Y \). \( X \sim \mathcal{N} (\mu, \sigma^2) \) indicates a Gaussian random variable with mean \( \mu \) and variance \( \sigma^2 \).

\begin{figure}[t!]
\centering
\includegraphics[width=3.5in]{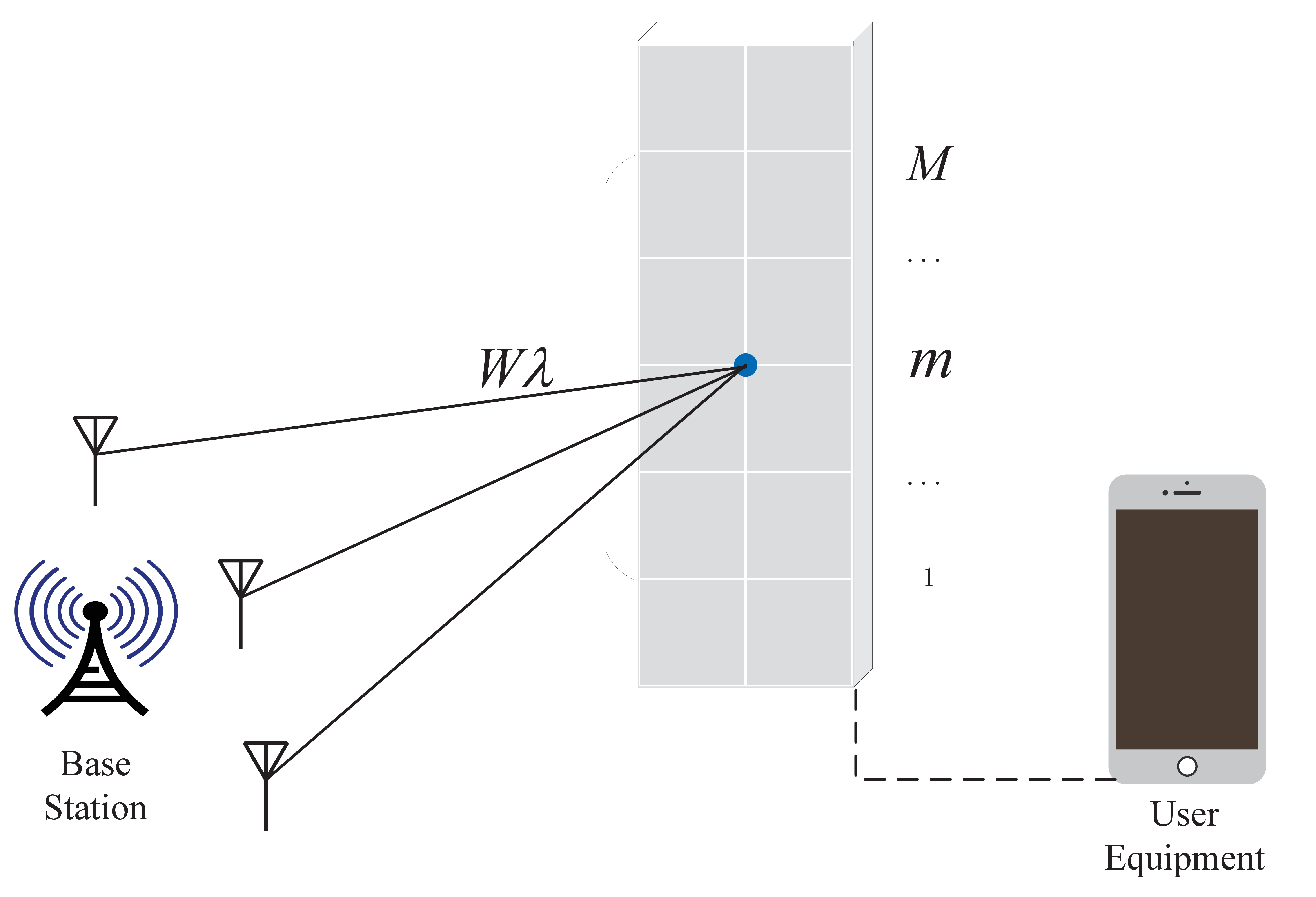}
\caption{The MISO-FAS model.}\label{system1}
\vspace{-3mm}
\end{figure}

\section{SYSTEM AND CHANNEL MODEL}\label{sec:model}

\subsection{MISO-FAS}\label{sec:misomodel}
This subsection considers a wireless communication system, as shown in Fig.~\ref{system1}, where the BS is equipped with $N$ antennas, and the UE employs a FAS with $M$ ports. 
%In this FAS configuration, the ports are distributed on a linear space of $W\lambda$, where $W$ represents the normalized physical size of the FAS, and $\lambda$ denotes the wavelength.

We model the small-scale fading vector ${\mathbf{h}_m \in \mathbb{C}^{N \times 1}}$ using Rayleigh fading channels as follows:
\begin{equation}\label{channelvector}
{{\bf{h}}_m} = {\left[ {h_{m1}}, \ldots, {h_{mN}} \right]}^T,
\end{equation}
where $h_{mn}$ denotes the channel coefficient between the $n$-th antenna at the BS and the $m$-th port at the UE. Under this model, the amplitude of the channel gain, $|h_{mn}|$, follows a Rayleigh distribution, and its PDF is given by:
\begin{equation}\label{Rayleigh dis}
{p_{\left| {{h_{{mn}}}} \right|}}(r) = \frac{{2r}}{{{\sigma ^2}}}{e^{ - \frac{{{r^2}}}{{{\sigma ^2}}}}}, \quad r \ge 0,
\end{equation}
where $\sigma^2=\mathbb{E}[|h_{mn}|^2]$ is the average power of the channel gain.

Maximum-ratio transmission (MRT) is employed at the BS, for which the precoding vector $\mathbf{w}_m \in \mathbb{C}^{1\times N}$ is given by:
\begin{equation}\label{mrt}
  {\bf{w}}_m = \frac{{{\mathbf{h}_m^H}}}{{\left\| \mathbf{h}_m \right\|}}.
\end{equation}

We model the received signal at the $m$-th port with additive white Gaussian noise (AWGN) as follows:
\begin{equation}\label{AWGN1}
y_m = \sqrt{P}\mathbf{w}_m\mathbf{h}_m x + {\eta},
\end{equation}
where $P$ is the total transmit power, $x$ represents the transmitted signal with $\mathbb{E}[x]=0$, and $\eta\sim \mathcal{N}_c (0, \eta_0^2)$ is the noise.

The received SNR $\gamma_m$ can be expressed as:
\begin{equation}\label{SNR}
\begin{aligned}
\gamma_m &= \frac{{P{{\left|\mathbf{w}_m\mathbf{h}_m \right|}^2}}}{\eta_0^2 } = \frac{{P{{\left|\frac{{{\mathbf{h}_m^H}}}{{\left\| \mathbf{h}_m \right\|}}\mathbf{h}_m \right|}^2}}}{\eta_0^2 } \\
 &= \frac{{P{{\left\| \mathbf{h}_m \right\|}^2}}}{\eta_0^2 } = \frac{{P\sum\limits_{n = 1}^N {|h_{mn}|^2} }}{\eta_0^2 } \equiv  \frac{{PX_{m} }}{\eta_0^2 },
 \end{aligned}
\end{equation}
where $X_{m}\buildrel \Delta \over =\sum\limits_{n = 1}^N {|h_{mn}|^2} $ is the precoding gain.
%, and $\bar{\gamma}\buildrel \Delta \over =\frac{P}{\eta_0^2 }$ is the average receive SNR.

The average single-branch SNR $\bar{\gamma}$ at each port is given by:
\begin{equation}\label{avsnr}
\bar{\gamma}=   \frac{{P\mathbb{E}[|h_{mn}|^2]}}{\eta_0^2 }  =\frac{{P\sigma^2 }}{\eta_0^2 } .
\end{equation}

The system selects the port with the highest instantaneous received SNR to be activated. Therefore, the resultant SNR of the proposed system is given by:
\begin{equation}\label{snrfas}
\gamma_{{\rm{FAS}}} = \max \left\{ \gamma_1, \gamma_2, \ldots, \gamma_M \right\}.
\end{equation}

%We adopt an equally correlated model, where the correlation factor between any two sub-channels on the same side is identical. The receive correlation factors are denoted by $\rho_{n_1n_2}=\rho^2$ for $n_1 \ne n_2$. 
Due to the compact physical placement of the ports, the channel gains are spatially correlated. We adopt an equally correlated model, where the correlation factor between any two sub-channels on the same side is identical. The receive correlation factors are denoted by $\rho_{m_1m_2}=\rho^2$ for $m_1 \ne m_2$. With this correlation model, the channel responses can be modeled as:
\begin{equation}\label{channel model-0}
\left\{ \begin{aligned}
{h_{1n}} &= \sigma {x_{0n}} + j\sigma {y_{0n}}\\
{h_{mn}} &= \sigma \left( {\sqrt {1 - \rho^2} {x_{mn}} + {\rho }{x_{0n}}} \right) \\
&+ j\sigma \left( {\sqrt {1 - \rho^2} {y_{mn}} + {\rho}{y_{0n}}} \right), \\
&\quad\quad \text{for } n=1, 2, \dots, N, \text{ and } m=2, \ldots, M,
\end{aligned}\right.
\end{equation}
where $x_{0n}, y_{0n}, x_{mn}, y_{mn}$ are mutually independent Gaussian random variables (RVs) with zero mean and a variance of $1/2$.
Since the channel amplitude is assumed to be Rayleigh-distributed, its square follows an exponential distribution, and the sum of squares follows a chi-square distribution. 
%Thus, the precoding gain at each port can be parameterized as:
%\begin{equation}\label{channelgain}
%\left\{ \begin{array}{l}
%{X_1} = \sum\limits_{n = 1}^N {{{\left| {{h_{1n}}} \right|}^2}} \\
%{X_m} = \sum\limits_{n = 1}^N {{{\left| {{h_{mn}}} \right|}^2}}, \quad \text{for } m = 2, \ldots ,M.
%\end{array} \right.
%\end{equation}
%where $\chi_0, \chi_m$ are independent noncentral chi-square distribution RVs.
%, and $\rho '$ is the correlation factor between  $\chi_0$ and $\chi_m$.

\subsection{DUAL-FAS}
This subsection considers a wireless communication system, as shown in Fig.~\ref{system2}, where the BS employs a FAS with $M_T$ ports, and the UE employs a FAS with $M_R$ ports.

As in the MISO-FAS model, there exists spatial correlation among the subchannels due to the close spacing between ports. Thus, we also adopt the equally correlated model, where the correlation factor between any two ports on the same side is identical. The receive correlation factor between any two ports is denoted by $\rho_1^2$, while the transmit correlation factor is denoted by $\rho_2^2$. Thus, the receiving and transmitting correlation matrices, $\mathbf{R}_R$ and $\mathbf{R}_T$, are given by:
\begin{equation}\label{rr}
\begin{array}{l}
{{\bf{R}}_R} = \left[ {\begin{array}{*{20}{c}}
1&{\rho _1^2}& \cdots &{\rho _1^2}\\
{\rho _1^2}&1& \cdots &{\rho _1^2}\\
 \vdots & \vdots & \ddots & \vdots \\
{\rho _1^2}&{\rho _1^2}& \cdots &1
\end{array}} \right],\\
{{\bf{R}}_T} = \left[ {\begin{array}{*{20}{c}}
1&{\rho _2^2}& \cdots &{\rho _2^2}\\
{\rho _2^2}&1& \cdots &{\rho _2^2}\\
 \vdots & \vdots & \ddots & \vdots \\
{\rho _2^2}&{\rho _2^2}& \cdots &1
\end{array}} \right].
\end{array}
\end{equation}

The channel matrix is modeled as:
\begin{equation}\label{AWGN2}
{\bf{H}}= {\mathbf{R}}_R^{\frac{1}{2}}{\bf{H}}_w{{\mathbf{R}}_T^{\frac{1}{2}}},
\end{equation}
where ${\bf{H}}_w$ is the independent and identically distributed (i.i.d.) channel matrix. The matrix ${\bf{H}}$ is explicitly given by:
\begin{equation}\label{H}
\mathbf{H} = \begin{bmatrix}
{h_{11}} & \cdots &{{h_{1M_T}}}\\
\vdots & \ddots  & \vdots \\
{{h_{M_R1}}}  & \cdots &{{h_{M_RM_T}}}
\end{bmatrix},
\end{equation}
where the element $h_{ij}$ denotes the complex channel coefficient from the $j$-th transmit port (TP) to the $i$-th receive port (RP). The proposed system comprises a total of $D = M_R M_T$ sub-channels.

The overall spatial correlation matrix of the proposed system is expressed as the Kronecker product of the transmit and receive correlation matrices:
\begin{equation}\label{R}
{{\bf{R}}}={{\bf{R}}_T^{\mathrm{T}}} \otimes{{\bf{R}}_R},
\end{equation}
where $\mathbf{R}$ is a $M_RM_T\times M_RM_T$ matrix, and the element $\rho_{ij}\in\mathbf{R}$ can be obtained as
\begin{equation}\label{rho}
{\rho _{ij}}=\left\{ \begin{array}{l}
1\quad{i_1} = {i_2},{j_1} = {j_2},\\
{\rho _1}\quad{i_1} = {i_2},{j_1} \ne {j_2},\\
{\rho _2}\quad{i_1} \ne {i_2},{j_1} = {j_2},\\
{\rho _1}{\rho _2}\quad{i_1} \ne {i_2},{j_1} \ne {j_2}.
\end{array} \right.
\end{equation}

With that, the complex channel gain from the $j$-th transmit port to the $i$-th receive port can be parameterized as:
\begin{equation}\label{para2}
\left\{ \begin{array}{l}
{h_{11}} = {z_0},\\
{h_{i1}} = \sqrt {1 - \rho _1^2} {z_{i1}} + {\rho _1}{z_0},\\
{h_{1j}} = \sqrt {1 - \rho _2^2} {z_{1j}} + {\rho _2}{z_0},\\
{h_{ij} }= \sqrt {1 - \rho _1^2\rho _2^2} {z_{ij}} + {\rho _1}{\rho _2}{z_0},\\
~for~i= 2, \ldots, M_R, j = 2, \ldots,M_T, 
\end{array} \right.
\end{equation}
where $z_0, z_{i1}, z_{1j}, z_{ij}$ are independent complex Gaussian RVs with zero mean and variance of $\sigma^2$.

\begin{figure}[t!]
\centering
\includegraphics[width=3.5in]{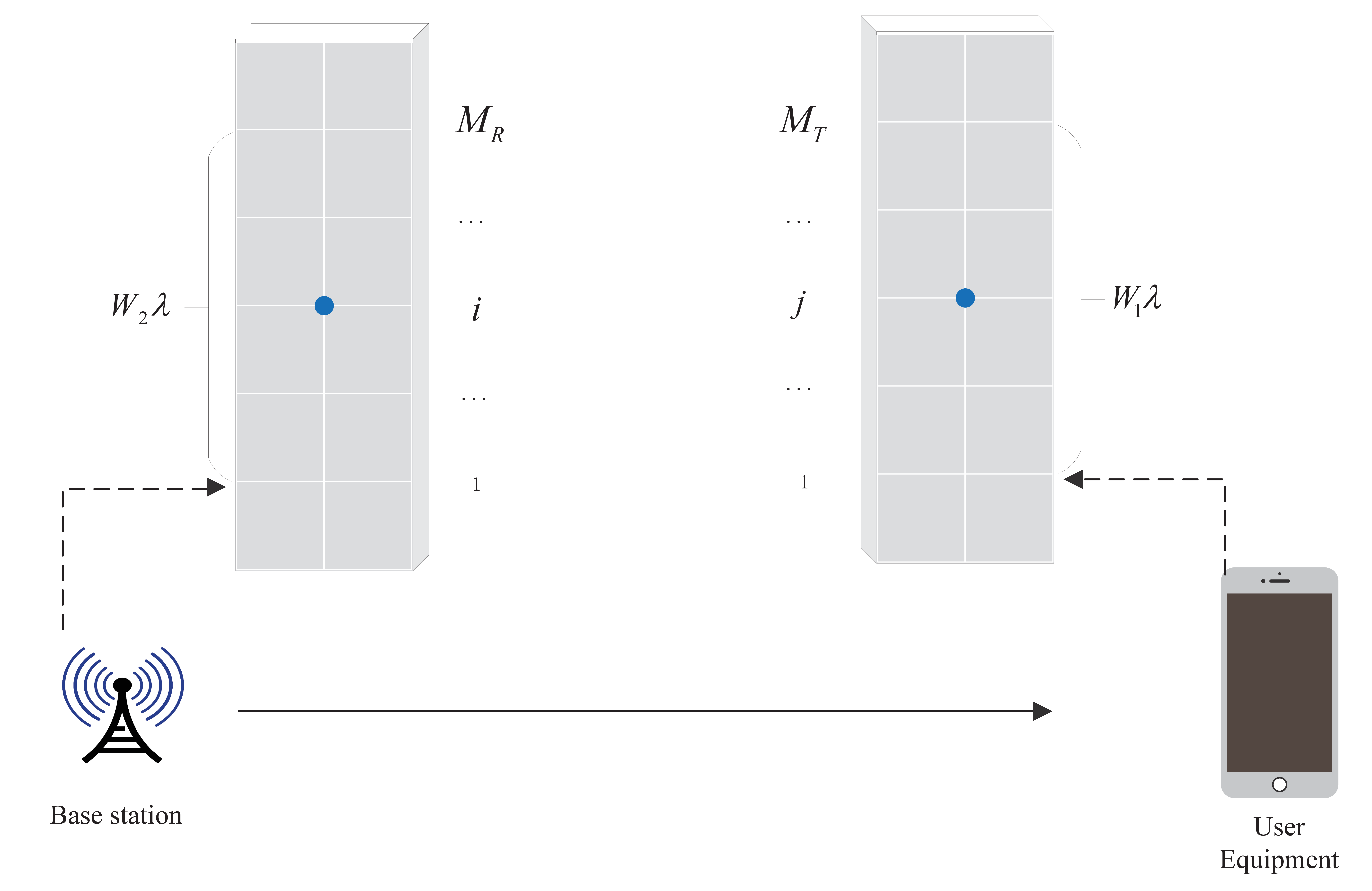}
\caption{The Dual-FAS model.}\label{system2}
\vspace{-3mm}
\end{figure}

\section{PERFORMANCE ANALYSIS}\label{sec:analysis}

In this section, we analyze the system performance of the proposed MISO-FAS and Dual-FAS models, respectively.

\subsection{MISO-FAS}
We first investigate the performance of the proposed MISO-FAS model, where the BS is equipped with $N$ antennas and the UE is equipped with a FAS with $M$ ports.

We assume the FAS always activates the port corresponding to the largest instantaneous precoding gain. The resulting precoding gain for the MISO-FAS, denoted by $X_{\rm FAS}$, is therefore given by:
\begin{equation}\label{channel model}
X_{\rm FAS} = \max \left\{ {X_1},{X_2}, \ldots, {X_M} \right\}.
\end{equation}
To analyze the statistical properties of $X_{\rm FAS}$, we first derive the joint PDF of ${X_1}, {X_2}, \ldots, {X_M}$.

\begin{lemma}\label{Lemma1:joint pdf}
Under the assumption of Rayleigh fading and the equally correlated channel model, the joint PDF of the precoding gains ${X_1}, {X_2}, \ldots, {X_M}$ across all ports of the MISO-FAS is given by:
\begin{equation}\label{jpdf}
\begin{aligned}
& p_{{X_1},{X_2}, \ldots, {X_M}}\left( {z_1}, z_2, \dots, z_M \right) =\frac{{{z_1}^{N-1}}}{{{\sigma ^{2N}\left(N-1\right)!}}}{e^{ - \frac{{{z_1}}}{{{\sigma ^2}}}}}\times\\
& \prod\limits_{m = 2}^M { \frac{1}{\sigma^2}{{e^{ - \frac{{{z_m} + {\rho ^2}{z_1}}}{{{\sigma ^2}(1 - {\rho ^2})}}}}{{\left( {\frac{{{z_m}(1 - {\rho ^2})}}{{{\rho ^2}{z_1}}}} \right)}^{\frac{N-1}{2}}}{I_{N-1}}\left( \sqrt{\frac{{ {4{\rho ^2}{z_1}{z_m}} }}{{1 - {\rho ^2}}}} \right)}} .
\end{aligned}
\end{equation}
\begin{proof}
See Appendix A.
\end{proof}
\end{lemma}

\begin{lemma}\label{Lemma2:joint CDF}
Under the assumption of Rayleigh fading and the equally correlated channel model, the joint CDF of the precoding gains ${X_1}, {X_2}, \ldots, {X_M}$ is given by:
\begin{equation}\label{cdf}
\begin{aligned}
& {F_{{X_1},{X_2}, \ldots, {X_M}}}(z) = \Pr(X_1 \le z, \ldots, X_M \le z) \\
& \quad = \int_0^{z} \frac{{{t_1}^{N-1}}}{{{\sigma ^{2N}\left(N-1\right)!}}}{e^{ - \frac{{{t_1}}}{{{\sigma ^2}}}}} \\
& \quad \quad \times {\left( {1 - {Q_N}\left( {\sqrt {\frac{{2\rho ^2t_1}}{{{\sigma ^2}(1 - {\rho ^2})}} } ,\sqrt {\frac{{2z}}{{{\sigma ^2}(1 - {\rho ^2})}}} } \right)} \right)^{M - 1}}d{t_1}.
\end{aligned}
\end{equation}
\begin{proof}
See Appendix B.
\end{proof}
\end{lemma}

We define the outage event as
 \begin{equation}\label{ope}
   \left\{ {\frac{{{{\gamma}_{\rm FAS}}}}{{{{\bar{\gamma}}}}}   < {\gamma _{th}}} \right\}=   \left\{{X_{\rm FAS}}<{\gamma_{th}}\right\},
 \end{equation}
 where $\gamma _{th}$ represents the SNR threshold.

\begin{theorem}\label{Theorem1:outage probability}
As defined in~\eqref{ope}, for the MISO-FAS model under Rayleigh fading, the OP can be written as:
\begin{equation}\label{op}
\begin{aligned}
&{P_{\text{out}}}({\gamma _{th}}) = \Pr(X_{\text{FAS}} <  \gamma_{th}) \\
&=\int_0^{{\gamma_{th}}} \frac{{{t_1}^{N-1}}}{{{\sigma ^{2N}\left(N-1\right)!}}}{e^{ - \frac{{{t_1}}}{{{\sigma ^2}}}}} \\
 &\times {\left( {1 - {Q_N}\left( {\sqrt {\frac{{2\rho ^2t_1}}{{{\sigma ^2}(1 - {\rho ^2})}} } ,\sqrt {\frac{{2}{\gamma_{th}}}{{{\sigma ^2}(1 - {\rho ^2})}}} } \right)} \right)^{M - 1}}d{t_1}.
\end{aligned}
\end{equation}

\begin{proof}
Substituting \(z_1=\dots=z_M=\gamma_{th}\) into the joint CDF in~\eqref{cdf} yields the OP and completes the proof.
\end{proof}
\end{theorem}

The OP expression in \eqref{op} is highly complex and given in an open form. Thus, we further analyze its closed-form upper and lower bounds.

\begin{corollary}\label{out_up1}
Under the assumption of Rayleigh fading and the equally correlated channel model, an upper bound on the OP of MISO-FAS for a given SNR threshold $\gamma_{th}$ is given by:
\begin{equation}\label{outup}
\begin{aligned}
P_{\text{out}}\left( {{\gamma _{th}}} \right) \le {{\tilde P}_{\text{out}}}\left( {{\gamma _{th}}} \right) &= \frac{1}{{\left( {N - 1} \right)!}}\gamma \left({N, {\gamma _{th}}} \right)\\
& \quad \times{\left( {1 - {e^{ - \frac{{2  {\gamma _{th}}}}{{{\sigma ^2}(1 - {\rho ^2})}}}}} \right)^{M - 1}}.
\end{aligned}
\end{equation}
\begin{proof}
See Appendix C.
\end{proof}
\end{corollary}

\begin{corollary}\label{out_low1}
Under the assumption of Rayleigh fading and the equally correlated channel model, a lower bound on the OP of MISO-FAS for a given SNR threshold $\gamma_{th}$ is given by:
\begin{equation}\label{outlow}
\begin{aligned}
P_{\text{out}}(\gamma_{th}) &\ge {{\hat P}_{\text{out}}}(\gamma_{th}) \\
&= \frac{1}{{\left( {N - 1} \right)!}}\gamma \left({N, \frac{\gamma _{th}}{\bar{\gamma}}} \right)\\
& \times{\left( {1 - {Q_N}\left( {\sqrt {\frac{{\rho ^2{\gamma _{th}}}}{{{\sigma ^2}(1 - {\rho ^2})}}} ,\sqrt {\frac{2  {\gamma _{th}}}{{{\sigma ^2}(1 - {\rho ^2})}}} } \right)} \right)^{M - 1}}.
\end{aligned}
\end{equation}
\begin{proof}
See Appendix D.
\end{proof}
\end{corollary}

\begin{proposition}\label{proposition1: do}
\emph{According to \textbf{Corollary~\ref{out_up1}}, the diversity order for the MISO-FAS model in the high-SNR regime is given by:}
\begin{equation}\label{dio}
{d_{MISO}} =  - \mathop {\lim }\limits_{\frac{1}{{{\gamma _{th}}}} \to \infty } \frac{{\log {P_{\text{out}}}}}{{\log \frac{1}{{{\gamma _{th}}}}}} \approx N+M.
\end{equation}
\begin{proof}
To derive the diversity order, we substitute the upper bound on the OP from \eqref{outup} into \eqref{dio}, which yields:
\begin{equation}\label{dio1}
\begin{aligned}
{d_{MISO}}& =  - \mathop {\lim }\limits_{\frac{1}{{{\gamma _{th}}}} \to \infty } \frac{{\log {P_{\text{out}}}}}{{\log \frac{1}{{{\gamma _{th}}}}}}\\
 &>  - \mathop {\lim }\limits_{\frac{1}{{{\gamma _{th}}}} \to \infty } \frac{{\log \left( {\gamma \left({N,{\gamma _{th}}} \right){{\left( {1 - {e^{ - \frac{{2{\gamma _{th}}}}{{{\sigma ^2}(1 - {\rho ^2})}}}}} \right)}^{M - 1}}} \right)}}{{\log \frac{1}{{{\gamma _{th}}}}}}\\
 & \approx  - \mathop {\lim }\limits_{\frac{1}{{{\gamma _{th}}}} \to \infty } \frac{{\log \left( {\frac{{{\gamma _{th}}^N}}{N}} \right) + \left( {M - 1} \right)\log \left( {{ \frac{{2{\gamma _{th}}}}{{{\sigma ^2}(1 - {\rho ^2})}}} } \right)}}{{\log \frac{1}{{{\gamma _{th}}}}}}\\
& \approx N+M-1 \approx N+M.
\end{aligned}
\end{equation}
Then, the proof is complete.
\end{proof}
\end{proposition}

\begin{remark}\label{remarkop}
Based on the results in~\eqref{dio}, the diversity order for the MISO-FAS model is approximately determined by the sum of the antennas and ports, $N+M$. This observation implies that the outage behavior can be continuously enhanced by simply utilizing more ports. This significant enhancement is achievable under the condition that the channels are not fully correlated (i.e., $|\rho| < 1$).

\end{remark}

\subsection{DUAL-FAS}
This subsection analyze the performance of the proposed Dual-FAS, where both the BS and the UE employ FAS with $M_T$ and $M_R$ ports, respectively.

The small-scale fading follows Rayleigh fading, which implies that the channel gain, $|h_{ij}|$, follows a Rayleigh distribution. Consequently, the channel power gain, $|h_{ij}|^2$, follows an exponential distribution. The instantaneous SNR of the $(i, j)$-th sub-channel, denoted by $\gamma_{mn}$, is thus given by:
\begin{equation}
    \gamma_{ij} = \frac{P}{\eta_0^2} |h_{ij}|^2 ={\bar \gamma '}\frac{|h_{ij}|^2}{\mathbb{E}[|h_{ij}|^2]} = \bar{\gamma} |h_{ij}|^2,
\end{equation}
where $\bar{\gamma}' = \frac{P \mathbb{E}[|h_{ij}|^2]}{\eta_0^2}$ is the average SNR per sub-channel. The system maximizes the received SNR by always selecting the TP-RP pair that provides the maximum channel gain. The resulting channel gain, denoted as $|h_{\text{FAS}}|$, is given by:
\begin{equation}\label{channel model}
|h_{\text{FAS}}| = \max \left\{ |h_{11}|, |h_{12}|, \ldots, |h_{M_RM_T}| \right\}.
\end{equation}

To analyze the statistical properties of $|h_{\text{FAS}}|$, we first derive the joint PDF of all channel gains, i.e., $|h_{11}|, |h_{12}|, \ldots, |h_{M_RM_T}|$.

\begin{lemma}\label{Lemma3:joint pdf}
Under the assumption of Rayleigh fading and the Kronecker product correlation model, the joint PDF of the channel gains of Dual-FAS, $\left| {{h_{11}}} \right|,\left| {{h_{12}}} \right|, \ldots ,\left| {{h_{M_RM_T}}} \right|$, is given by
\begin{equation}\label{jpdf2}
\begin{aligned}
&p_{\left| {{h_{11}}} \right|,\left| {{h_{12}}} \right|, \ldots ,\left| {{h_{M_RM_T}}} \right|}\left( {{r_{11}}, \ldots {r_{M_RM_T}}} \right) = \frac{{2{r_{11}}}}{{{\sigma ^2}}}\exp \left( { - \frac{{r_{11}^2}}{{{\sigma ^2}}}} \right)\\
 &\times \prod\limits_{i = 2}^{M_R} {\frac{{2{r_{i1}}}}{{{\sigma ^2}(1 - \rho _1^2)}}} \exp \left( { - \frac{{r_{i1}^2 + \rho _1^2r_{11}^2}}{{{\sigma ^2}(1 - \rho _1^2)}}} \right){I_0}\left( {\frac{{2{\rho _1}{r_{11}}{r_{i1}}}}{{{\sigma ^2}(1 - \rho _1^2)}}} \right)\\
& \times \prod\limits_{j = 2}^{M_T}{\frac{{2{r_{1n}}}}{{{\sigma ^2}(1 - \rho _2^2)}}} \exp \left( { - \frac{{r_{1j}^2 + \rho _2^2r_{11}^2}}{{{\sigma ^2}(1 - \rho _2^2)}}} \right){I_0}\left( {\frac{{2{\rho _2}{r_{11}}{r_{1j}}}}{{{\sigma ^2}(1 - \rho _2^2)}}} \right)\\
& \times \prod\limits_{i = 2}^{M_R} {\prod\limits_{j = 2}^{M_T} {\frac{{2{r_{ij}}}}{{{\sigma ^2}(1 - \rho _1^2\rho _2^2)}}} } \exp \left( { - \frac{{r_{ij}^2 + \rho _1^2\rho _2^2r_{11}^2}}{{{\sigma ^2}(1 - \rho _1^2\rho _2^2)}}} \right)\\
&\times{I_0}\left( {\frac{{2{\rho _1}{\rho _2}{r_{ij}}{r_{11}}}}{{{\sigma ^2}(1 - \rho _1^2\rho _2^2)}}} \right).
\end{aligned}
\end{equation}
\begin{proof}
See Appendix E.
\end{proof}
\end{lemma}

\begin{lemma}\label{Lemma4:joint cdf}
Under the assumption of Rayleigh fading and the Kronecker product correlation model, the joint CDF of the channel power gains of Dual-FAS, $\left| {{h_{11}}} \right|^2,\left| {{h_{12}}} \right|^2, \ldots ,\left| {{h_{M_RM_T}}} \right|^2$, is given by
\begin{equation}\label{jcdf2}
\begin{aligned}
&{F_{\left| {{h_{11}}} \right|^2,\left| {{h_{12}}} \right|^2, \ldots ,\left| {{h_{M_RM_T}}} \right|^2}}\left( {{r_{11}}, \ldots {r_{M_RM_T}}} \right) = \int_0^{{r_{11}}} {{{\mathop{\rm e}\nolimits} ^{ - t}}} \\
& \times \prod\limits_{i = 2}^{M_R} {\left( {1 - {Q_1}\left( {\sqrt {\frac{{2\rho _1^2}}{{1 - \rho _1^2}}t} ,\sqrt {\frac{{2{r_{i1}}}}{{1 - \rho _1^2}}} } \right)} \right)} \\
& \times \prod\limits_{j= 2}^{M_T} {\left( {1 - {Q_1}\left( {\sqrt {\frac{{2\rho _2^2}}{{1 - \rho _2^2}}t} ,\sqrt {\frac{{2{r_{1j}}}}{{1 - \rho _2^2}}} } \right)} \right)} \\
& \times \prod\limits_{i = 2}^{M_R} {\prod\limits_{j = 2}^{M_T} {\left( {1 - {Q_1}\left( {\sqrt {\frac{{2\rho _2^2\rho _1^2t}}{{1 - \rho _1^2\rho _2^2}}} ,\sqrt {\frac{{2{r_{ij}}}}{{1 - \rho _1^2\rho _2^2}}} } \right)} \right)} } dt.
\end{aligned}
\end{equation}

\begin{proof}
See Appendix F.
\end{proof}
\end{lemma}

The outage event occurs when the normalized received SNR falls below the SNR threshold $\gamma_{th}$:
\begin{equation}\label{ope2}
\left\{ {\frac{{{\gamma _{{\rm{FAS}}}}}}{{\bar \gamma }} = |{h_{{\rm{FAS}}}}{|^2} < {\gamma _{th}}} \right\}.
\end{equation}

\begin{theorem}\label{Theorem2:op}
Based on the definition of the outage event in~\eqref{ope2}, the OP of the Dual-FAS under Rayleigh fading channels is expressed as:
\begin{equation}\label{op2}
\begin{aligned}
&P_{\text{out}}\left( {{\gamma _{th}}} \right) = \int_0^{{\gamma _{th}}} {{{\mathop{\rm e}\nolimits} ^{ - t}}} \\
& \times \prod\limits_{i = 2}^{M_R} {\left( {1 - {Q_1}\left( {\sqrt {\frac{{2\rho _1^2t}}{{1 - \rho _1^2}}} ,\sqrt {\frac{{2{\gamma _{th}}}}{{1 - \rho _1^2}}} } \right)} \right)} \\
& \times \prod\limits_{j = 2}^{M_T} {\left( {1 - {Q_1}\left( {\sqrt {\frac{{2\rho _2^2t}}{{1 - \rho _2^2}}} ,\sqrt {\frac{{2{\gamma _{th}}}}{{1 - \rho _2^2}}} } \right)} \right)} \\
& \times \prod\limits_{i = 2}^{M_R} {\prod\limits_{j = 2}^{M_T} {\left( {1 - {Q_1}\left( {\sqrt {\frac{{2\rho _2^2\rho _1^2t}}{{1 - \rho _1^2\rho _2^2}}} ,\sqrt {\frac{{2{\gamma _{th}}}}{{1 - \rho _1^2\rho _2^2}}} } \right)} \right)} } dt.
\end{aligned}
\end{equation}

\begin{proof}
Substituting \(r_{11}=\dots=r_{M_RM_T}=\gamma_{th}\)  into the joint CDF in~\eqref{jcdf2} yields the OP and completes the proof.
\end{proof}
\end{theorem}

We consider the following special case as a sanity check.
\begin{corollary}

For the special case of $M_T=1$, the Dual-FAS model reduces to the Rx-SISO-FAS model analyzed in \cite{FAS}. To align with the setup in \cite{FAS}, we assume the correlation follows the Jakes model, the receive correlation coefficient is given by:
\begin{equation}\label{rou2}
\begin{aligned}
{\rho _i} & = {J_0}\left( {2\pi \frac{{\left( {i - 1} \right)W\lambda }}{{\left( {M_R - 1} \right)\lambda }}} \right)\\
 &= {J_0}\left( {\frac{{2\pi (i - 1)}}{{M_R - 1}}W} \right),\;{\rm{for}}\;i = 1,2, \ldots ,M_R,
\end{aligned}
\end{equation}
where $J_0(\cdot)$ is the zero-order Bessel function of the first kind. By substituting \eqref{rou2} and $M_T=1$ into \eqref{op2}, the resulting OP is given by:

\begin{equation}\label{rxFAS}
\begin{aligned}
&P_{\text{out}}\left( {{\gamma _{th}}} \right) = \int_0^{{\gamma _{th}}} {{e^{ - t}}}\\
& \times \prod \limits_{i = 2}^{M_R} {\left[ {1 - {Q_1}\left( {\sqrt {\frac{{2\rho _i^2t}}{{1 - \rho _i^2}} } ,\sqrt {\frac{{2}}{{1 - \rho _i^2}}} \sqrt {{\gamma _{th}}} } \right)} \right]} dt.
\end{aligned}
\end{equation}

This result is mathematically consistent with the OP expression of the Rx-SISO-FAS under Rayleigh fading as presented in \cite{FAS}.

\end{corollary}

\begin{corollary}\label{out_up2}
Based on the Rayleigh fading assumption and the Kronecker product correlation model, a closed-form upper bound on the OP of the Dual-FAS model is derived as:
\begin{equation}\label{outup2}
\begin{aligned}
&P_{\text{out}}\left( {{\gamma _{th}}} \right) \le {{\tilde P}_{\text{out}}}\left( {{\gamma _{th}}} \right) \\
&= \left( {1 - {e^{ - {\gamma _{th}}}}} \right){\left( {1 - {e^{ - \frac{{{\gamma _{th}}}}{{1 - \rho _1^2}}}}} \right)^{M_R-1}} \\
&\quad \times {\left( {1 - {e^{ - \frac{{{\gamma _{th}}}}{{1 - \rho _2^2}}}}} \right)^{M_T-1}}{\left( {1 - {e^{ - \frac{{{\gamma _{th}}}}{{1 - \rho _1^2\rho _2^2}}}}} \right)^{\left( {M_R-1} \right)\left( {M_T-1} \right)}}.
\end{aligned}
\end{equation}
\begin{proof}
See Appendix G.
\end{proof}
\end{corollary}

\begin{corollary}\label{out_low2}
Based on the Rayleigh fading assumption and the Kronecker product correlation model, a closed-form lower bound on the OP of the Dual-FAS model is derived as:
\begin{equation}\label{outlow2}
\begin{aligned}
&P_{\text{out}}\left( {{\gamma _{th}}} \right) \ge {{\hat P}_{\text{out}}}\left( {{\gamma _{th}}} \right) = \left( {1 - {e^{ - {\gamma _{th}}}}} \right) \\
&\times {\left( {1 - {Q_1}\left( {\sqrt {\frac{{2\rho _1^2{\gamma _{th}}}}{{1 - \rho _1^2}}} ,\sqrt {\frac{{2{\gamma _{th}}}}{{1 - \rho _1^2}}} } \right)} \right)^{M_R - 1}} \\
& \times {\left( {1 - {Q_1}\left( {\sqrt {\frac{{2\rho _2^2{\gamma _{th}}}}{{1 - \rho _2^2}}} ,\sqrt {\frac{{2{\gamma _{th}}}}{{1 - \rho _2^2}}} } \right)} \right)^{M_T - 1}} \\
& \times {\left( {1 - {Q_1}\left( {\sqrt {\frac{{2\rho _1^2\rho _2^2{\gamma _{th}}}}{{1 - \rho _1^2\rho _2^2}}} ,\sqrt {\frac{{2{\gamma _{th}}}}{{1 - \rho _1^2\rho _2^2}}} } \right)} \right)^{\left( {M_R - 1} \right)\left( {M_T- 1} \right)}}.
\end{aligned}
\end{equation}
\begin{proof}
See Appendix H.
\end{proof}
\end{corollary}

\begin{proposition}\label{proposition2: do2}
\emph{In the high-SNR regime, the diversity order of the Dual-FAS model can be derived from the OP upper bound in \textbf{Corollary~\ref{out_up2}} as:}
\begin{equation}\label{do2}
{d_{Dual}} =  - \mathop {\lim }\limits_{\frac{1}{{{\gamma _{th}}}} \to \infty } \frac{{\log {{P}_{out}}}}{{\log \frac{1}{{{\gamma _{th}}}}}} = M_RM_T.
\end{equation}
\begin{proof}
By substituting the upper bound of OP from \eqref{outup2} into the definition of diversity order in \eqref{do2}, we have the diversity order \eqref{dio1}, which is on the top of the next page.
\begin{figure*}
\begin{equation}\label{dio1}
\begin{aligned}
{d_{Dual}}& =  - \mathop {\lim }\limits_{\frac{1}{{{\gamma _{th}}}} \to \infty } \frac{{\log {P_{out}}}}{{\log \frac{1}{{{\gamma _{th}}}}}}\\
 & >  - \mathop {\lim }\limits_{\frac{1}{{{\gamma _{th}}}} \to \infty } \frac{{\log \left( {\left( {1 - {e^{ - {\gamma _{th}}}}} \right){{\left( {1 - {e^{ - \frac{{{\gamma _{th}}}}{{1 - \rho _1^2}}}}} \right)}^{M_R - 1}}{{\left( {1 - {e^{ - \frac{{{\gamma _{th}}}}{{1 - \rho _2^2}}}}} \right)}^{M_T - 1}}{{\left( {1 - {e^{ - \frac{{{\gamma _{th}}}}{{1 - \rho _1^2\rho _2^2}}}}} \right)}^{\left( {M_R - 1} \right)\left( {M_T - 1} \right)}}} \right)}}{{\log \frac{1}{{{\gamma _{th}}}}}}\\
  &\approx  - \mathop {\lim }\limits_{\frac{1}{{{\gamma _{th}}}} \to \infty } \frac{{\log \left( {{\gamma _{th}}} \right) + \left( {M_R - 1} \right)\log \left( {\frac{{{\gamma _{th}}}}{{1 - \rho _1^2}}} \right) + \left( {M_T - 1} \right)\log \left( {\frac{{{\gamma _{th}}}}{{1 - \rho _2^2}}} \right) + \left( {M_R - 1} \right)\left( {M_T - 1} \right)\log \left( {\frac{{{\gamma _{th}}}}{{1 - \rho _1^2\rho _2^2}}} \right)}}{{\log \frac{1}{{{\gamma _{th}}}}}}\\
& \approx M_RM_T.
\end{aligned}
\end{equation}
\hrulefill
\end{figure*}
Thus, the proof is concluded.
\end{proof}
\end{proposition}

\begin{remark}\label{remarkop2}
Based on the result in~\eqref{do2}, the Dual-FAS model can achieve a diversity order of $M_RM_T$. This indicates a significant performance enhancement compared to conventional single-antenna systems, as the system leverages spatial diversity from both the transmit and receive sides.
\end{remark}

\begin{figure}[t!]
\centering
\includegraphics[width=3.5in]{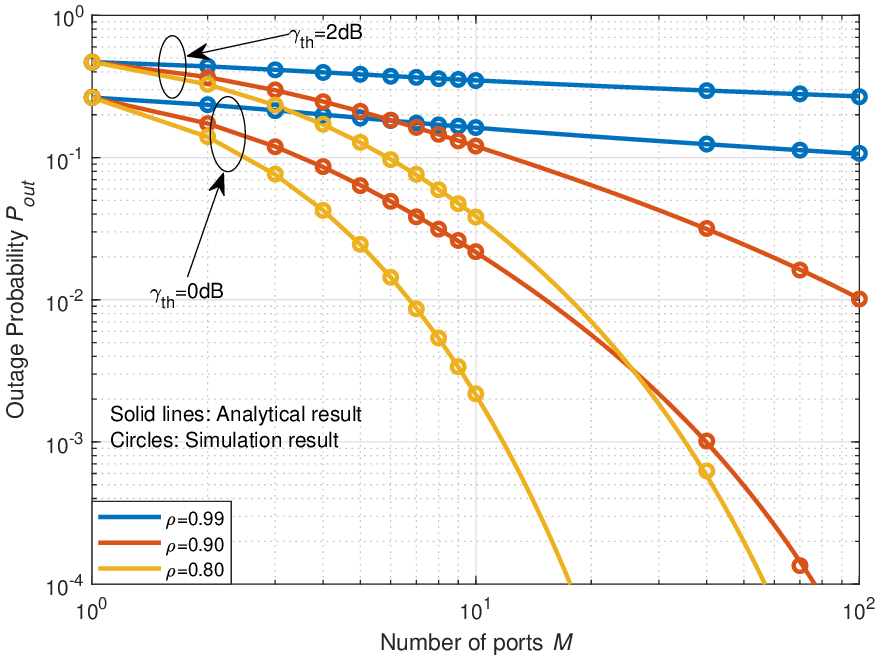}
\caption{The OP against the number of ports $M$ for 2-antenna MISO-FAS}\label{ports}
\vspace{-3mm}
\end{figure}

\section{NUMERICAL RESULTS}\label{sec:simu}
This section provides simulation results to assess the performance of MISO-FAS and Dual-FAS in Rayleigh fading channels. All analytical findings are validated via Monte Carlo simulations, averaged over $ 1 \times 10^6 $ independent trials.

\subsection{IMPACT OF THE NUMBER OF PORTS ON THE OP OF MISO-FAS}

In Fig.~\ref{ports}, the OP of MISO-FAS is analyzed as a function of the number of ports under different correlation factors and SNR thresholds. The analytical results (solid lines) show a close match with the Monte Carlo simulations (circles), validating the accuracy of our theoretical model. The OP decreases with the increase in the number of ports under different SNR thresholds. Moreover, the OP curve is observed to exhibit a steeper slope with a larger number of ports. In addition, weaker spatial correlation leads to a faster decrease in OP. When the correlation factor \(\rho = 0.99\), increasing the number of ports by 100 results in a reduction in OP of less than an order of magnitude. In contrast, when \(\rho = 0.8\), the OP decreases by more than two orders of magnitude ($\gamma_{th}=0$ dB). This implies that lower spatial correlation allows the system to exploit higher spatial diversity gains from the additional ports.

\begin{figure}[t!]
\centering
\includegraphics[width=3.5in]{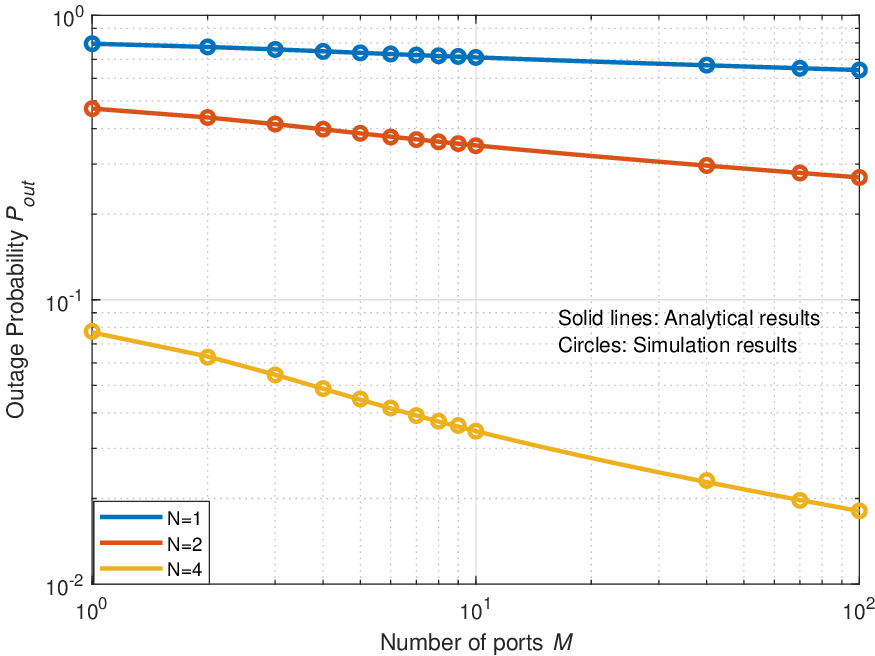}
\caption{The OP against the number of ports $M$ for $N$-antenna MISO-FAS when the correlation factor $\rho=0.99$.}\label{Ns}
\vspace{-3mm}
\end{figure}

\subsection{IMPACT OF THE NUMBER OF ANTENNAS ON THE OP OF MISO-FAS}
In Fig.~\ref{Ns}, we investigate how the OP of the MISO-FAS varies with the number of ports under different numbers of antennas $N$. As expected, for any given number of antennas, the OP consistently decreases as the number of ports $M$ increases. A key observation, however, is the substantial performance gain afforded by increasing the number of antennas. For instance, at $M=10$ ports, increasing the number of antennas from $N=1$ to $N=4$ reduces the OP from approximately 0.7 to $4 \times 10^{-2}$. 

Furthermore, the curves with a larger number of antennas exhibit steeper slopes, indicating that the MISO-FAS system benefits more significantly from additional ports when equipped with more antennas.
In conclusion, these results demonstrate that even under high spatial correlation, deploying multiple antennas significantly enhances the reliability gains provided by the FAS.

\begin{figure}[t!]
\centering
\includegraphics[width=3.5in]{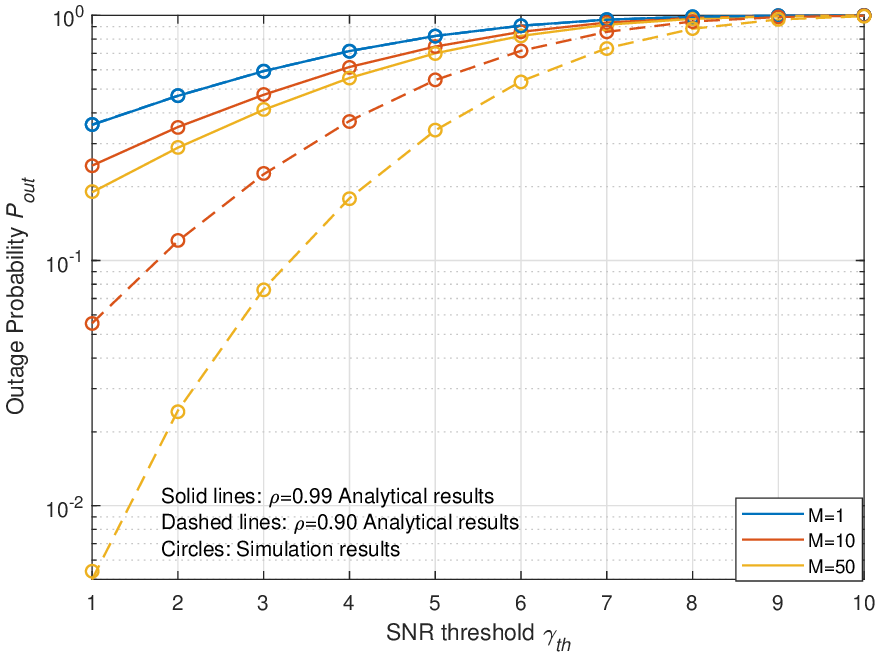}
\caption{The OP against the SNR threshold $\gamma_{th}$ for $2$-antenna MISO-FAS when the correlation factor $\rho=0.99$.}\label{SNR}
\vspace{-3mm}
\end{figure}

\subsection{IMPACT OF SNR THRESHOLD ON THE OP OF MISO-FAS} 
In Fig.~\ref{SNR}, we investigate the OP of a 2-antenna MISO-FAS versus the SNR threshold $\gamma_{th}$, considering different numbers of ports $M$ and two channel correlation factors $\rho$. When the number of ports is $M = 1$, the fluid antenna reduces to a fixed-position antenna, and the correlation coefficient has no impact on system performance. Consequently, the two blue curves in Fig.~\ref{SNR} ($\rho = 0.99$ and $\rho = 0.90$) completely overlap. As anticipated, for any given configuration, the OP increases as the SNR threshold increases. 
Additionally, reducing the correlation factor leads to a sharp decrease in OP. For instance, with $M=50$ ports and an SNR threshold of $\gamma_{th}=3$, the OP falls from approximately 0.4 (at $\rho=0.99$) to below 0.1 (at $\rho=0.90$). Notably, the slope of the OP curves becomes steeper as the number of ports increases, which verifies the analysis of \textbf{Remark~\ref{remarkop}}. 

\begin{figure}[t!]
\centering
\includegraphics[width=3.5in]{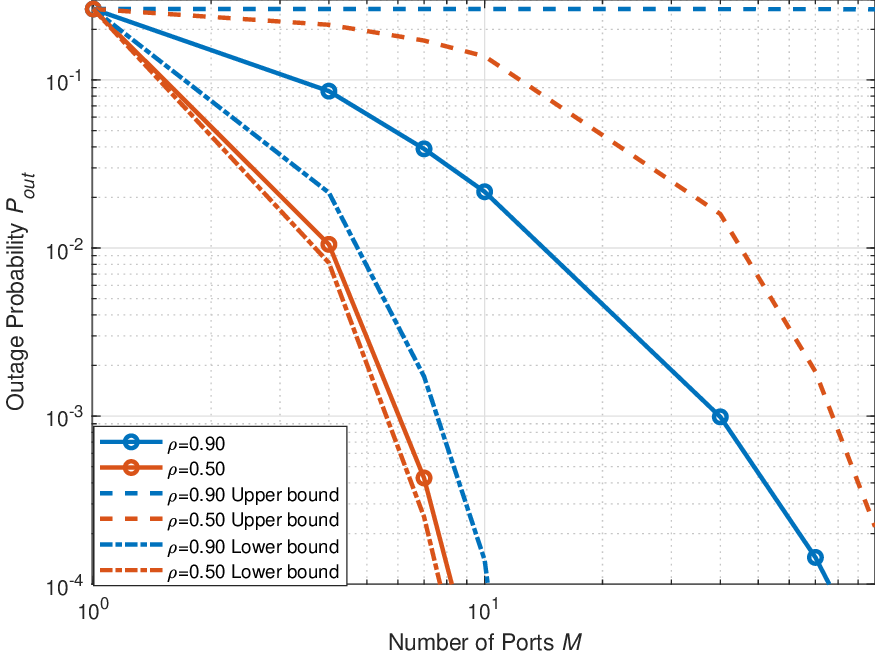}
\caption{The OP and its bounds for 2-antenna MISO-FAS when the SNR threshold $\gamma_{th}=0$dB under different correlation factor $\rho$.}\label{MISOUL}
\vspace{-3mm}
\end{figure}

\subsection{UPPER AND LOWER BOUNDS ON THE OP OF MISO-FAS}
 In Fig.~\ref{MISOUL}, we present the upper and lower bounds for the OP of a 2-antenna MISO-FAS, plotted against the number of ports for different correlation factors at a fixed SNR threshold of $\gamma_{th}=0$ dB. The solid lines with circles represent the exact OP, which serves as a benchmark. The primary observation is that for both correlation scenarios, the exact OP curves are strictly bounded by their corresponding upper and lower bounds, thereby validating the accuracy of our theoretical derivations. 
 Furthermore, the system with lower correlation $\rho=0.50$ exhibits a significantly steeper slope, indicating a superior outage performance compared to the high-correlation case. It is also noteworthy that the bounds for the lower correlation scenario are exceptionally tight, closely tracking the exact performance, whereas the bounds for the high-correlation case are less tight. This finding suggests that the lower bound serves as a close approximation for the OP under low channel correlation. In contrast, the upper bound provides a more conservative performance estimate across all conditions.

\begin{figure}[t!]
\centering
\includegraphics[width=3.5in]{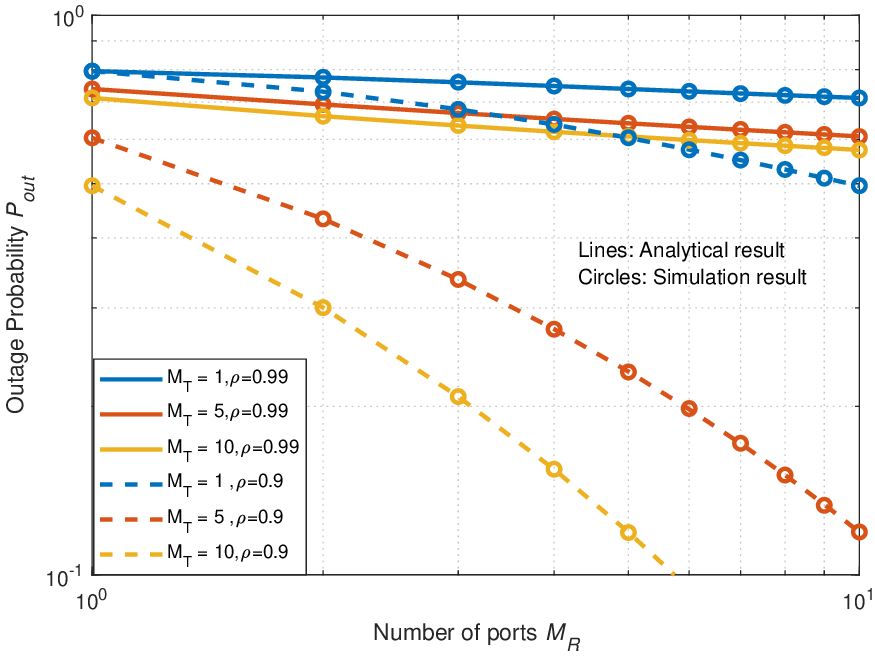}
\caption{The OP against the number of RPs $M_R$ for Dual-FAS}\label{dualm}
\vspace{-3mm}
\end{figure}

\subsection{IMPACT OF THE NUMBER OF RPS ON THE OP OF DUAL-FAS} 
In Fig.~\ref{dualm}, we study the OP of the Dual-FAS as a function of the number of RPs, considering different numbers of TPs and correlation factors. The analytical results show a close match with the Monte Carlo simulations, affirming the accuracy of our theoretical model. It is observed that a significant decrease in correlation from $\rho=0.99$ (solid lines) to $\rho=0.9$ (dashed lines) leads to a reduction in OP. For example, with $M_T=10$ and $M_R=3$, the OP decreases from approximately 0.6 to 0.12 as the correlation factor reduces. Moreover, for a fixed correlation, increasing the number of TPs $M_T$ and RPs $M_R$ consistently enhances system performance. The results further indicate that the OP curve for $\rho=0.9$ is steeper than that for $\rho=0.99$. This is because lower spatial correlation enables the system to exploit higher diversity gains by increasing the number of RPs.

\begin{figure}[t!]
\centering
\includegraphics[width=3.5in]{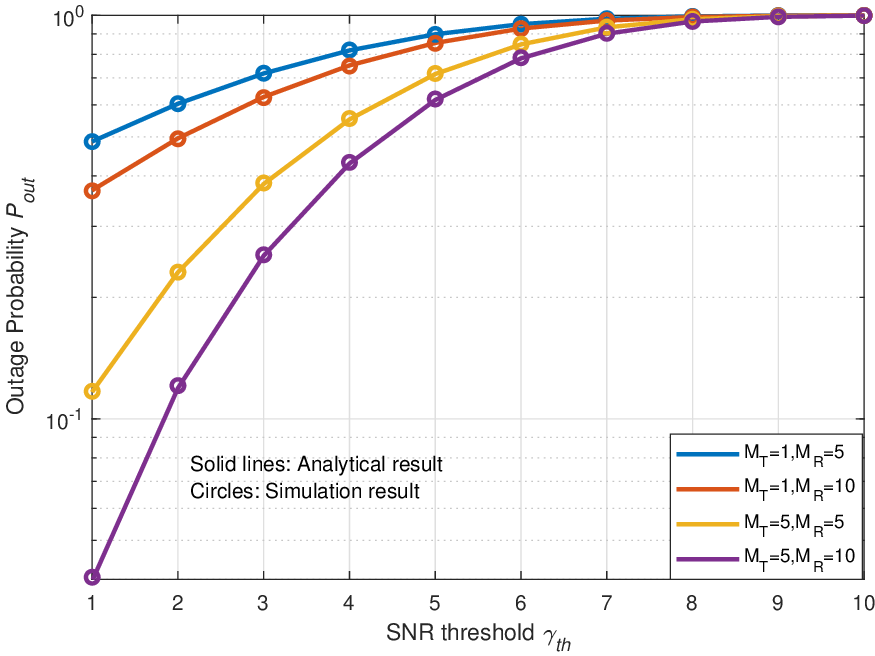}
\caption{The OP of Dual-FAS against the SNR threshold $\gamma_{th}$ with different correlation factor $\rho=0.9$.}\label{dualSNR}
\vspace{-3mm}
\end{figure}

\subsection{IMPACT OF SNR THRESHOLD ON THE OP OF DUAL-FAS} 
Fig.~\ref{dualSNR} illustrates the OP of the Dual-FAS versus the SNR threshold for different numbers of ports, under a fixed correlation of $\rho=0.9$. As expected, the OP increases monotonically with the SNR threshold for all configurations. The slope of the curves becomes steeper for configurations with more ports, confirming the analysis presented in \textbf{Remark~\ref{remarkop2}}. The key insight is the significant performance improvement achieved by deploying FAS at both the transmitter and receiver sides. For example, with the SNR threshold of $\gamma_{th}=4$ dB, increasing the number of TPs from $M_T=1$ to $M_T=5$ (for $M_R=10$) reduces the OP from approximately 0.75 to 0.45. Similarly, the system with more ports consistently outperforms the one with fewer ports. The effect is most pronounced in the $M_T=5$, $M_R=10$ configuration, which exhibits the best performance. In conclusion, the key insight from this figure is that increasing the number of ports provides a direct and effective approach to enhancing the performance.

\begin{figure}[t!]
\centering
\includegraphics[width=3.5in]{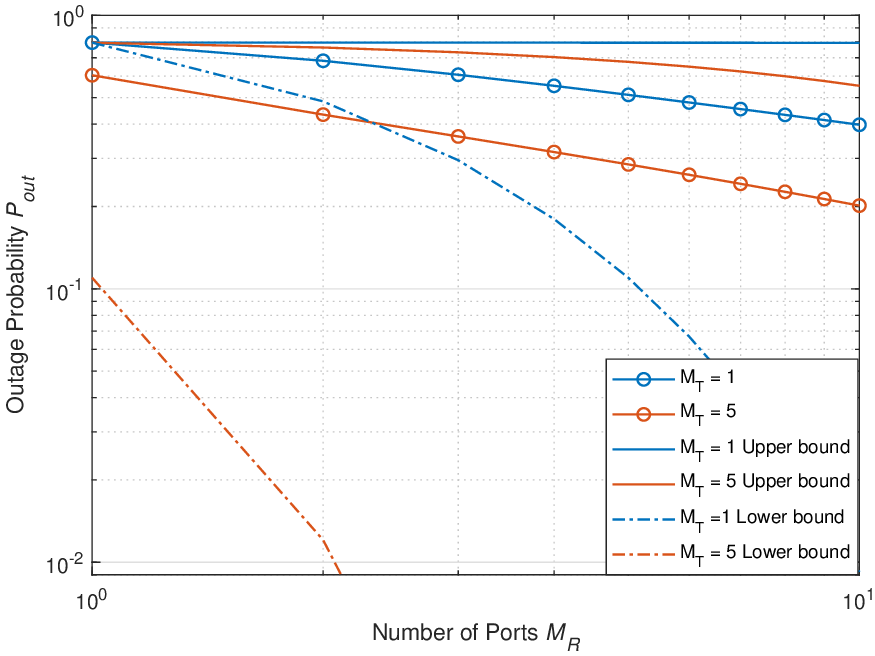}
\caption{The upper and lower bounds of the OP of Dual-FAS with the SNR threshold $\gamma_{th}=0$dB and correlation factor $\rho=0.99$.}\label{dualul}
\vspace{-3mm}
\end{figure}

\subsection{UPPER AND LOWER BOUNDS ON THE OP OF DUAL-FAS}
In Fig.~\ref{dualul}, we examine the upper and lower bounds on the OP of the Dual-FAS system as a function of the number of RPs, under a high correlation factor of $\rho = 0.99$ and a SNR threshold of $\gamma_{th} = 0$ dB. As the number of RPs increases, the lower bound effectively captures the steep diversity slope, particularly for the $M_T=5$ case. This indicates the potential for significant performance improvement, whereas the upper bound provides a conservative performance limit. In summary, this figure validates our analytical results by showing that the derived bounds accurately characterize the outage performance of the Dual-FAS system.

\begin{figure}[t!]
\centering
\includegraphics[width =3.5in]{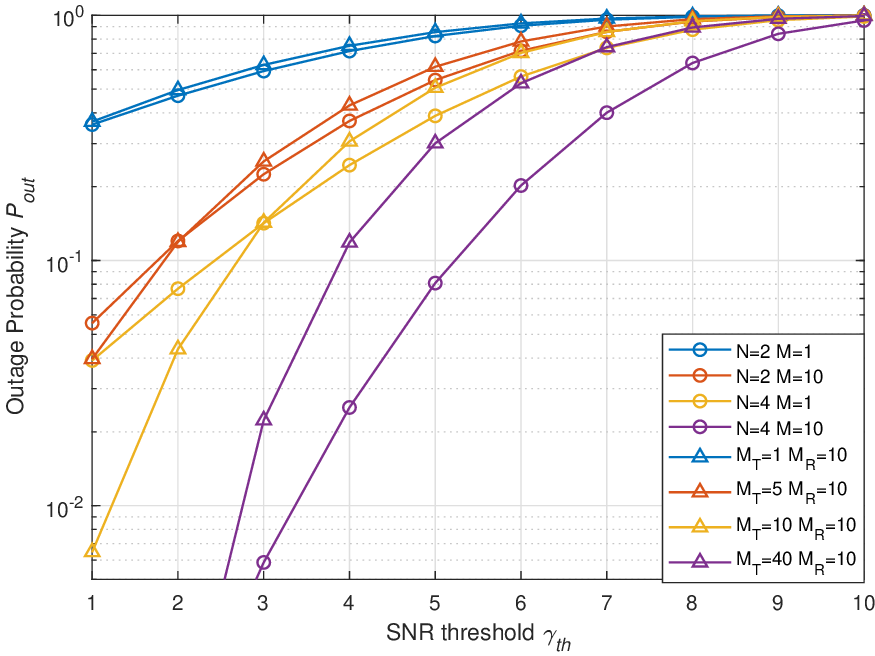}
\caption{The OP of MISO-FAS and Dual-FAS with different number of ports and antennas with correlation factor $\rho=0.9$.}\label{FASMRC}
\end{figure}

\subsection{PERFORMANCE COMPARISON OF MISO-FAS AND DUAL-FAS}
In Fig.~\ref{FASMRC}, we provide a comparative analysis of the OP for MISO-FAS (circles) and Dual-FAS (triangles) as a function of the SNR threshold, under a correlation factor of $\rho = 0.9$. As expected, for all configurations, the OP increases as the SNR threshold increases. The central observation is the clear performance enhancement gained by increasing the number of ports in both models. Given the same number of ports at the receiver ($M_R=M$), the MISO-FAS model exhibits better performance in the low-SNR regime. However, the key advantage of Dual-FAS over MISO-FAS is its ability to leverage an increased number of ports for performance gains with fewer RF chains. Moreover, for configurations of the two models that achieve comparable performance in the low-SNR region (curves of the same color), the slope of the Dual-FAS curve is steeper than that of MISO-FAS as the SNR threshold decreases, indicating that in the high-SNR regime, the performance of Dual-FAS surpasses that of MISO-FAS.

\section{CONCLUSIONS}\label{sec:conclude}
This paper presented a comprehensive performance analysis of MISO-FAS and Dual-FAS systems over Rayleigh fading channels. For both models, we derived the exact OP expressions, along with their corresponding upper and lower bounds. It was found that the diversity order of MISO-FAS is determined by the sum of the number of antennas and ports $M+N$, while that of Dual-FAS is dictated by the product of the numbers of ports at both ends $M_RM_T$. Furthermore, comparative numerical results demonstrated that the system performance of both models improves significantly as spatial correlation decreases. Finally, while this study highlights the substantial gains offered by FAS, many aspects of this emerging technology remain open for future investigation.

\numberwithin{equation}{section}
\section*{Appendix~A: Proof of Lemma~\ref{Lemma1:joint pdf}}\label{Appendix:As}
\renewcommand{\theequation}{A.\arabic{equation}}
\setcounter{equation}{0}
Conditioned on $x_0,y_0$, $\left|h_{mn}\right|$ follows a Rician distribution, and we have
\begin{equation}\label{conpdf}
\begin{aligned}
p_{\left| {{h_{mn}}} \right|\left|{{x_0},{y_0}}\right.}({r_{mn}}) &= \frac{{2{r_{mn}}}}{{{\sigma ^2}(1 - {\rho ^2})}}{e^{ - \frac{{r_{mn}^2 + {\rho ^2} {\left( {x_0^2 + y_0^2} \right)} }}{{{\sigma ^2}(1 - {\rho ^2})}}}} \\
&\times {I_0}\left( {\frac{{2\rho \sqrt{ {x_0^2 + y_0^2} }{r_{mn}}}}{{{\sigma ^2}(1 - {\rho ^2})}}} \right).
\end{aligned}
\end{equation}
According to \eqref{conpdf}, $\left|h_{mn}\right|$ is only correlated to $\left|h_{1n}\right|$. By performing a variable transformation, the conditioned PDF of $\left|h_{mn}\right|^2$ can be obtained as
\begin{equation}\label{conpdfsquare}
\begin{aligned}
{p_{{{\left| {{h_{mn}}} \right|}^2}\left| {{{\left| {{h_{1n}}} \right|}^2}} \right.}}({r_{mn}}) &= \frac{1}{{{\sigma ^2}(1 - {\rho ^2})}}{e^{ - \frac{{{r_{mn}} + {\rho ^2}{r_{1n}}}}{{{\sigma ^2}(1 - {\rho ^2})}}}}\\
& \times {I_0}\left( {\frac{{2\left( {\sqrt {{\rho ^2}{r_{1n}}{r_{mn}}} } \right)}}{{{\sigma ^2}(1 - {\rho ^2})}}} \right).
\end{aligned}
\end{equation}
Based on \eqref{conpdfsquare}, conditioned on $\left|h_{1n}\right|^2$, $\left|h_{mn}\right|^2$ follows a noncentral Chi-square distribution with DoF $k=2$ and noncentral factor ${\lambda _n} = \frac{{2{\rho ^2}r_{1n}^2}}{{{\sigma ^2}(1 - {\rho ^2})}}$.

Thus, conditioned on $\left|h_{11}\right|^2, \dots, \left|h_{1N}\right|^2$, $X_m$ follows noncentral Chi-square distribution with DoF $k=2N$ and noncentral factor $\lambda_{MRT}=\frac{{2{\rho ^2}}}{{{\sigma ^2}(1 - {\rho ^2})}}\sum\limits_{n = 1}^N {r_{1n}^2} $, and we have
\begin{equation}\label{2pdf}
\begin{aligned}
{p_{{X_m}\left| {{{\left| {{h_{11}}} \right|}^2}, \ldots ,{{\left| {{h_{1N}}} \right|}^2}} \right.}}\left( {{z_m}} \right) &=\frac{1}{\sigma^2}{e^{ - \frac{{{\frac{2z_m}{\sigma^2}} + {\lambda _{MRT}}}}{2}}}{\left( {\frac{{{2z_m}}}{{{\sigma^2\lambda _{MRT}}}}} \right)^{\frac{{N - 1}}{2}}}\\
&\times{I_{N - 1}}\left( {\sqrt {\frac{{2\lambda _{MRT}}{z_m}}{\sigma^2}} } \right).
\end{aligned}
\end{equation}

According to \eqref{2pdf}, $X_m$ only depends on $X_1$.
Since  $|h_{11}|$ and $|h_{12}|$ both follow a Rayleigh distribution, \( X_1 \) follows a gamma distribution with DoF $k=2$. Thus, the PDF of $X_1$ can be written as:
\begin{equation}\label{A4}
{f_{{Z_1}}}\left( {{z_1}} \right) = \frac{{{z_1}^{N-1}}}{{{\sigma ^{2N}\left(N-1\right)!}}}{e^{ - \frac{{{z_1}}}{{{\sigma ^2}}}}}.
\end{equation}
Based on \eqref{2pdf} and \eqref{A4}, the joint PDF of ${X_1},{X_2},...,{X_M}$ is given by:
\begin{equation}\label{jpdf0}
\begin{aligned}
&{p_{{X_1},{X_2},...,{X_M}}}\left( {{z_1},{z_2}, \ldots {z_M}} \right)\\
 &={f_{{Z_1}}}\left( {{z_1}} \right)\prod\limits_{m = 2}^M {{p_{{X_m}\left| {{X_1}} \right.}}} \left( {{z_m}} \right)\\
& =\frac{{{z_1}^{N-1}}}{{{\sigma ^{2N}\left(N-1\right)!}}}{e^{ - \frac{{{z_1}}}{{{\sigma ^2}}}}}\times\\
& \prod\limits_{m = 2}^M {\left( \frac{{e^{ - \frac{{{z_m} + {\rho ^2}{z_1}}}{{{\sigma ^2}(1 - {\rho ^2})}}}}}{\sigma^2}{{{\left( {\frac{{{z_m}(1 - {\rho ^2})}}{{{\rho ^2}{z_1}}}} \right)}^{\frac{N-1}{2}}}{I_{N-1}}\left( \sqrt{\frac{{ {4{\rho ^2}{z_1}{z_m}} }}{{1 - {\rho ^2}}}} \right)} \right)} .
\end{aligned}
\end{equation}
Thus, the proof is completed.

\vspace{-0.1in}
\numberwithin{equation}{section}
\section*{Appendix~B: Proof of Lemma~\ref{Lemma2:joint CDF}}\label{Appendix:Bs}
\renewcommand{\theequation}{B.\arabic{equation}}
\setcounter{equation}{0}

According to the joint PDF in \eqref{jpdf}, the joint CDF can be derived as
\begin{equation}\label{B1}
\begin{aligned}
&{F_{{X_1},{X_2},...,{X_M}}}({z_1},{z_2},...,{z_M})\\
 &= P\left( {{X_1} < {z_1},{X_2} < {z_2}, \ldots ,{X_M} < {z_M}} \right)\\
 &= \int_0^{{z_1}} { \cdots \int_0^{{z_M}} {{p_{{X_1},{X_2},...,{X_M}}}\left( {{t_1},{t_2}, \ldots {t_M}} \right)d{t_1},{t_2}, \ldots {t_M}} } \\
 &=\int_0^{{z_1}} { \cdots \int_0^{{z_M}} \frac{{{t_1}}}{{{\sigma ^2}}}{e^{ - \frac{{{t_1}}}{{{\sigma ^2}}}}}}\times\\
 &\prod\limits_{m = 2}^M {\left( {{e^{ - \frac{{{t_m} + {\rho ^2}{t_1}}}{{{\sigma ^2}(1 - {\rho ^2})}}}}{{\left( {\frac{{{t_m}{\sigma ^2}(1 - {\rho ^2})}}{{{\rho ^2}{t_1}}}} \right)}^{\frac{1}{2}}}{I_{N-1}}\left( {\frac{{\sqrt {4{\rho ^2}{t_1}{t_m}} }}{{{\sigma ^2}(1 - {\rho ^2})}}} \right)} \right)}\\
 &d{t_1}{t_2} \ldots {t_M} .
\end{aligned}
\end{equation}
The integral over the chi-square PDF inside the product operator has a closed-form solution in terms of the Marcum-Q function, given by \cite{marcumq}
\begin{equation}\label{B2}
F\left( z \right) = 1 - {Q_{\frac{k}{2}}}\left( {\sqrt {{\lambda _{MRT}}} ,\sqrt z } \right).
\end{equation}
By substituting \eqref{B2} into \eqref{B1}, we have
\begin{equation}\label{B3}
\begin{aligned}
&{F_{{X_1},{X_2},...,{X_M}}}({z_1},{z_2},...,{z_M})\\
 &= \int_0^{{z_1}} \frac{{{t_1}^{N-1}}}{{{\sigma ^{2N}\left(N-1\right)!}}}{e^{ - \frac{{{t_1}}}{{{\sigma ^2}}}}} \\
 &\times\prod\limits_{m = 2}^M {\left( {1 - {Q_{{N}}}\left( {\sqrt {\frac{{2{\rho ^2} {t_1} }}{{{\sigma ^2}(1 - {\rho ^2})}}} ,\sqrt {\frac{{2z_1}}{{{\sigma ^2}(1 - {\rho ^2})}}} } \right)} \right)} d{t_1}.
\end{aligned}
\end{equation}
Thus, the proof is completed.

\numberwithin{equation}{section}
\section*{Appendix~C: Proof of Corollary~\ref{out_up1}}\label{Appendix:Cs}
\renewcommand{\theequation}{C.\arabic{equation}}
\setcounter{equation}{0}
To derive the upper bound on the OP of MISO-FAS, our first step is to determine the OP reduction from the $M$-th additional port, as expressed in \eqref{opred}, which is on the top of the next page.
\begin{figure*}
\begin{equation}\label{opred}
\begin{aligned}
\Delta {P_{out}} &= P_{out}^{M - 1} - P_{out}^M\\
 &= \int_0^{{\gamma _{th}}} {{Q_N}\left( {\sqrt {\frac{2{\rho ^2{t_1}}}{{{\sigma ^2}(1 - {\rho ^2})}}} ,\sqrt {\frac{2{\gamma _{th}}}{{{\sigma ^2}(1 - {\rho ^2})}}} } \right)\frac{{{t_1}^{N-1}}}{{{\sigma ^{2N}\left(N-1\right)!}}}{e^{ - \frac{{{t_1}}}{{{\sigma ^2}}}}}{\left( {1 - {Q_N}\left( {\sqrt {\frac{2{\rho ^2{t_1}}}{{{\sigma ^2}(1 - {\rho ^2})}}} ,\sqrt {\frac{2{\gamma _{th}}}{{{\sigma ^2}(1 - {\rho ^2})}}} } \right)} \right)^{M - 2}}}d{t_1}.
\end{aligned}
\end{equation}
\hrulefill
\end{figure*}

The Cauchy-Schwarz lower bound of the Marcum-Q function is given by
\begin{equation}\label{C2}
{Q_v}\left( {a,b} \right) \ge {e^{ - \frac{{{b^2}}}{2}}},~\frac{a}{b} < 1.
\end{equation}
For $0<t<\gamma_{th}$, we have 
\begin{equation}\label{C3}
{Q_N}\left( {\sqrt {\frac{2{{\rho ^2}{t_1}}}{{{\sigma ^2}(1 - {\rho ^2})}}} ,\sqrt {\frac{{2{\gamma _{th}}}}{{{\sigma ^2}(1 - {\rho ^2})}}} } \right) \ge  {{e^{ - \frac{{2{\gamma _{th}}}}{{{\sigma ^2}(1 - {\rho ^2})}}}}} .
\end{equation}
By substituting \eqref{C3} into \eqref{opred}, we have
\begin{equation}\label{C4}
\Delta {P_{out}} = P_{out}^{M - 1} - P_{out}^M > {e^{ - \frac{{2{\gamma _{th}}}}{{{\sigma ^2}(1 - {\rho ^2})}}}}P_{out}^{M - 1}.
\end{equation}
Thus, the upper bound on the OP of MISO-FAS can be expressed as
\begin{equation}\label{C5}
\begin{aligned}
P_{out}^M &< \left( {1 - {e^{ - \frac{{2{\gamma _{th}}}}{{{\sigma ^2}(1 - {\rho ^2})}}}}} \right)P_{out}^{M - 1}\\
 &< {\left( {1 - {e^{ - \frac{{2{\gamma _{th}}}}{{{\sigma ^2}(1 - {\rho ^2})}}}}} \right)^{M - 1}}P_{out}^1.
\end{aligned}
\end{equation}
To complete the proof, we substitute the expression from \eqref{op} for the special case of $M=1$ into \eqref{C5}.

\section*{Appendix~D: Proof of Corollary~\ref{out_low1}}\label{Appendix:Ds}
\renewcommand{\theequation}{D.\arabic{equation}}
\setcounter{equation}{0}
To derive the lower bound on the OP of MISO-FAS, we first further express the OP reduction in \eqref{opred} as:
\begin{equation}\label{D1}
\begin{aligned}
&\Delta {P_{out}}= P_{out}^{ {{M} - 1} } - P_{out}^{{M}} \\
&=\left. {{Q_N}\left( {\sqrt {\frac{{2\rho ^2{t_1}}}{{{\sigma ^2}(1 - {\rho ^2})}}} ,\sqrt {\frac{{2{\gamma _{th}}}}{{{\sigma ^2}(1 - {\rho ^2})}}} } \right)P_{out}^{M-1}} \right|_0^{{\gamma _{th}}}\\
&- \int_0^{{\gamma _{th}}} {{{Q'}_N}\left( {\sqrt {\frac{{{2\rho ^2}{t_1}}}{{{\sigma ^2}(1 - {\rho ^2})}}} ,\sqrt {\frac{{2{\gamma _{th}}}}{{{\sigma ^2}(1 - {\rho ^2})}}} } \right)P_{out}^{M-1}}dt_1 \\
& \le {Q_N}\left( {\sqrt {\frac{{{2\rho ^2}{\gamma _{th}}}}{{{\sigma ^2}(1 - {\rho ^2})}}} ,\sqrt {\frac{2{\gamma _{th}}}{{{\sigma ^2}(1 - {\rho ^2})}}}} \right)P_{out}^{M-1}.
\end{aligned}
\end{equation}

After proving the second step in \eqref{D1} using the monotonicity of the Marcum Q-function and integration by parts from \cite{integral}, we can express the lower bound on the OP as:
\begin{equation}\label{D2}
\begin{aligned}
&P_{out}^M \ge \left( {1 - {Q_N}\left( {\sqrt {\frac{{2\rho ^2{\gamma _{th}}}}{{{\sigma ^2}(1 - {\rho ^2})}}} ,\sqrt {\frac{2{\gamma _{th}}}{{{\sigma ^2}(1 - {\rho ^2})}}}} \right)} \right)P_{out}^{M - 1}\\
 &= {\left( {1 - {Q_N}\left( {\sqrt {\frac{{2\rho^2{\gamma _{th}}}}{{{\sigma ^2}(1 - {\rho ^2})}}} ,\sqrt {\frac{2{\gamma _{th}}}{{{\sigma ^2}(1 - {\rho ^2})}}}} \right)} \right)^{M - 1}}P_{out}^1.
\end{aligned}
\end{equation}
As the final step, we substitute the expression from \eqref{op} for the case $M=1$ into \eqref{D2}, which completes the proof.

\section*{Appendix~E: Proof of Lemma~\ref{Lemma3:joint pdf}}\label{Appendix:Es}
\renewcommand{\theequation}{E.\arabic{equation}}
\setcounter{equation}{0}
We first adopt a similar approach to that in Appendix A to obtain the conditional PDF of $\left|h_{ij}\right|$, given by:
\begin{equation}\label{E1}
\begin{aligned}
{p_{\left| {{h_{ij}}} \right|\left| {{h_{11}}} \right.}}\left( {{r_{ij}}} \right)&=\frac{{2{r_{ij}}}}{{{\sigma ^2}(1 - \rho _1^2\rho _2^2)}} \exp \left( { - \frac{{r_{ij}^2 + \rho _1^2\rho _2^2r_{11}^2}}{{{\sigma ^2}(1 - \rho _1^2\rho _2^2)}}} \right)\\
&\times{I_0}\left( {\frac{{2{\rho _1}{\rho _2}{r_{ij}}{r_{11}}}}{{{\sigma ^2}(1 - \rho _1^2\rho _2^2)}}} \right).
\end{aligned}
\end{equation}
Thus, conditioned on $\left|h_{11}\right|$, the joint PDF of $\left|h_{12}\right|, \dots, \left|h_{M_RM_T}\right|$ is given by
\begin{equation}\label{E2}
\begin{aligned}
&{p_{\left| {{h_{12}}} \right| \ldots \left| {{h_{M_RM_T}}} \right|\left| {\left| {{h_{11}}} \right|} \right.}}\left( {{r_{12}}, \ldots {r_{M_RM_T}}} \right) = \\
& \times \prod\limits_{i = 2}^{M_R} {\frac{{2{r_{i1}}}}{{{\sigma ^2}(1 - \rho _1^2)}}} \exp \left( { - \frac{{r_{i1}^2 + \rho _1^2r_{11}^2}}{{{\sigma ^2}(1 - \rho _1^2)}}} \right){I_0}\left( {\frac{{2{\rho _1}{r_{11}}{r_{i1}}}}{{{\sigma ^2}(1 - \rho _1^2)}}} \right)\\
& \times \prod\limits_{j = 2}^{M_T} {\frac{{2{r_{1j}}}}{{{\sigma ^2}(1 - \rho _2^2)}}} \exp \left( { - \frac{{r_{1j}^2 + \rho _2^2r_{11}^2}}{{{\sigma ^2}(1 - \rho _2^2)}}} \right){I_0}\left( {\frac{{2{\rho _2}{r_{11}}{r_{1j}}}}{{{\sigma ^2}(1 - \rho _2^2)}}} \right)\\
& \times \prod\limits_{i = 2}^{M_R} {\prod\limits_{j = 2}^{M_T} {\frac{{2{r_{ij}}}}{{{\sigma ^2}(1 - \rho _1^2\rho _2^2)}}} } \exp \left( { - \frac{{r_{ij}^2 + \rho _1^2\rho _2^2r_{11}^2}}{{{\sigma ^2}(1 - \rho _1^2\rho _2^2)}}} \right)\\
&\times{I_0}\left( {\frac{{2{\rho _1}{\rho _2}{r_{ij}}{r_{11}}}}{{{\sigma ^2}(1 - \rho _1^2\rho _2^2)}}} \right).
\end{aligned}
\end{equation}
Finally, by multiplying \eqref{Rayleigh dis} with \eqref{E2}, the joint PDF of $\left|h_{11}\right|, \dots, \left|h_{M_RM_T}\right|$ in \eqref{jpdf2} is obtained, which completes the proof.

\section*{Appendix~F: Proof of Lemma~\ref{Lemma4:joint cdf}}\label{Appendix:Fs}
\renewcommand{\theequation}{F.\arabic{equation}}
\setcounter{equation}{0}
Based on the joint PDF in \eqref{jpdf2}, the joint CDF of ${\left| {{h_{11}}} \right|^2,\left| {{h_{12}}} \right|^2, \ldots ,\left| {{h_{M_RM_T}}} \right|^2}$ can be derived as
\begin{equation}\label{F1}
\begin{aligned}
&{F_{{\left| {{h_{11}}} \right|^2,\left| {{h_{12}}} \right|^2,\dots,\left| {{h_{M_RM_T}}} \right|^2}}}({r_{11}},{r_{12}},...,{r_{M_RM_T}})\\
 &= P\left( {\left| {{h_{11}}} \right|^2 < {r_{11}},\left| {{h_{12}}} \right|^2 < {r_{12}},\dots,\left| {{h_{M_RM_T}}} \right|^2< {r_{M_RM_T}}} \right)\\
 & = \int_0^{{r_{11}}} {{{\mathop{\rm e}\nolimits} ^{ - {t_{11}}}}} \\
 &\times \prod\limits_{i = 2}^{{M_R}} {\int_0^{{r_{i1}}} {\exp \left( { - \frac{{{t_{i1}} + \rho _1^2{t_{11}}}}{{1 - \rho _1^2}}} \right){I_0}\left( {\frac{{2{\rho _1}{r_{11}}{r_{i1}}}}{{1 - \rho _1^2}}} \right)} } d{t_{i1}}\\
& \times \prod\limits_{j = 2}^{{M_T}} {\int_0^{{r_{1j}}} {\exp \left( { - \frac{{{t_{ij}} + \rho _2^2{t_{11}}}}{{1 - \rho _2^2}}} \right){I_0}\left( {\frac{{2{\rho _2}{r_{11}}{r_{1j}}}}{{1 - \rho _2^2}}} \right)} } d{t_{1j}}\\
& \times \prod\limits_{i = 2}^{{M_R}} {\prod\limits_{j = 2}^{{M_T}} {\int_0^{{r_{ij}}} {\exp \left( { - \frac{{t_{ij} + \rho _1^2\rho _2^2t_{11}}}{{1 - \rho _1^2\rho _2^2}}} \right)} } } \\
&\times{I_0}\left( {\frac{{2{\rho _1}{\rho _2}{r_{ij}}{r_{11}}}}{{1 - \rho _1^2\rho _2^2}}} \right)d{t_{ij}}d{t_{11}}.
\end{aligned}
\end{equation}
The integral over the Rician PDF inside the product operator has a closed-form solution in terms of the Marcum-Q function, given by:
\begin{equation}\label{F2}
\begin{aligned}
F &= 1 - {Q_1}\left( {\frac{v}{{{\sigma _0}}},\frac{r}{{{\sigma _0}}}} \right)\\
&= 1 - {Q_1}\left( {\sqrt {\frac{{2{\rho _1^2\rho _2^2r_{11}^2}t_1^2}}{{{\sigma ^2}\left( {1 - \rho _1^2\rho _2^2} \right)}}} ,\sqrt {\frac{{2r}}{{{\sigma ^2}(1 -\rho _1^2\rho _2^2)}}} } \right).
\end{aligned}
\end{equation}
By substituting \eqref{F2} into \eqref{F1}, we have the joint CDF, which completes the proof.

\vspace{-0.1in}
\section*{Appendix~G: Proof of Corollary~\ref{out_up2}}\label{Appendix:Gs}
\renewcommand{\theequation}{G.\arabic{equation}}
\setcounter{equation}{0}
To derive the upper bound on the OP of Dual-FAS, we first adopt a similar approach in Appendix C to have the OP reduction in \eqref{G1}, which is on the bottom of the page.
\begin{figure*}[b]
\hrulefill
\begin{equation}\label{G1}
\begin{aligned}
\Delta {P_{out}} &= P_{out}^{\left( {{M_R} - 1} \right){M_T}} - P_{out}^{{M_R}{M_T}} \\
 &= \int_0^{{\gamma _{th}}} {{Q_1}\left( {\sqrt {\frac{{2\rho _2^2\rho _1^2t}}{{1 - \rho _1^2\rho _2^2}}} ,\sqrt {\frac{{2{\gamma _{th}}}}{{1 - \rho _1^2\rho _2^2}}} } \right){{\rm{e}}^{ - t}}}  \times \prod\limits_{i = 2}^{{M_R}} {\left( {1 - {Q_1}\left( {\sqrt {\frac{{2\rho _1^2t}}{{1 - \rho _1^2}}} ,\sqrt {\frac{{2{\gamma _{th}}}}{{1 - \rho _1^2}}} } \right)} \right)} \\
& \times \prod\limits_{j = 2}^{{M_T}} {\left( {1 - {Q_1}\left( {\sqrt {\frac{{2\rho _2^2t}}{{1 - \rho _2^2}}} ,\sqrt {\frac{{2{\gamma _{th}}}}{{1 - \rho _2^2}}} } \right)} \right)}  \times \prod\limits_{i = 2}^{{M_R} - 1} {\prod\limits_{j = 2}^{{M_T}} {\left( {1 - {Q_1}\left( {\sqrt {\frac{{2\rho _2^2\rho _1^2t}}{{1 - \rho _1^2\rho _2^2}}} ,\sqrt {\frac{{2{\gamma _{th}}}}{{1 - \rho _1^2\rho _2^2}}} } \right)} \right)} } dt.\\
\end{aligned}
\end{equation}
\end{figure*}

We can express the upper bound on the OP by substituting the lower bound of Marcum-Q function in \eqref{C2} into \eqref{G1} as
\begin{equation}\label{G2}
\begin{aligned}
&P_{out} \le \left( {1 - {e^{ - \frac{{2{\gamma _{th}}}}{{1 - \rho _1^2\rho _2^2}}}}} \right)P_{out}^{\left( {{M_R} - 1} \right){M_T}}\\
& = {\left( {1 - {e^{ - \frac{{2{\gamma _{th}}}}{{1 - \rho _1^2\rho _2^2}}}}} \right)}{\left( {1 - {e^{ - \frac{{{\gamma _{th}}}}{{1 - \rho _1^2}}}}} \right)^{{M_R} - 1}}\\
 &\times {\left( {1 - {e^{ - \frac{{{\gamma _{th}}}}{{1 - \rho _2^2}}}}} \right)^{{M_T} - 1}}{\left( {1 - {e^{ - \frac{{{\gamma _{th}}}}{{1 - \rho _1^2\rho _2^2}}}}} \right)^{\left( {{M_R} - 1} \right)\left( {{M_T} - 1} \right)}}P_{out}^1.
 \end{aligned}
\end{equation}
Finally, to complete the proof, we substitute the expression from \eqref{op2} for the special case of $M_R=M_T=1$ into \eqref{G2}.

\section*{Appendix~H: Proof of Corollary~\ref{out_low2}}\label{Appendix:Hs}
\renewcommand{\theequation}{H.\arabic{equation}}
\setcounter{equation}{0}
To derive the lower bound on the OP of Dual-FAS, we first further express the OP reduction in \eqref{G1} as:
\begin{equation}\label{H1}
\begin{aligned}
&\Delta {P_{out}}= \left. {{Q_1}\left( {\sqrt {\frac{{2\rho _2^2\rho _1^2t}}{{1 - \rho _1^2\rho _2^2}}} ,\sqrt {\frac{{2{\gamma _{th}}}}{{1 - \rho _1^2\rho _2^2}}} } \right)P_{out}^{\left( {{M_R} - 1} \right){M_T}}} \right|_0^{{\gamma _{th}}}\\
& - \int_0^{{\gamma _{th}}} {{{Q'}_1}\left( {\sqrt {\frac{{2\rho _2^2\rho _1^2t}}{{1 - \rho _1^2\rho _2^2}}} ,\sqrt {\frac{{2{\gamma _{th}}}}{{1 - \rho _1^2\rho _2^2}}} } \right)} P_{out}^{\left( {{M_R} - 1} \right){M_T}}dt\\
 &\le {Q_1}\left( {\sqrt {\frac{{2\rho _2^2\rho _1^2{\gamma _{th}}}}{{1 - \rho _1^2\rho _2^2}}} ,\sqrt {\frac{{2{\gamma _{th}}}}{{1 - \rho _1^2\rho _2^2}}} } \right)P_{out}^{\left( {{M_R} - 1} \right){M_T}}.
\end{aligned}
\end{equation}
The second step in~\eqref{H1} is established using the same approach as in Appendix D.
By appropriately rearranging the terms, we establish the lower bound of the OP in the form of the following recurrence relation:
\begin{equation}\label{H2}
\begin{aligned}
P_{out}^{{M_R}{M_T}} &\ge \left( 1 - {Q_1}\left( \sqrt{\frac{2\rho_2^2\rho_1^2{\gamma_{th}}}{1-\rho_1^2\rho_2^2}},\right.\right.\\
&\qquad\left.\left.\sqrt{\frac{2{\gamma_{th}}}{1-\rho_1^2\rho_2^2}} \right) \right) P_{out}^{(M_R-1)M_T}\\
&= {\left( 1 - {Q_1}\left( \sqrt{\frac{2\rho_1^2{\gamma_{th}}}{1-\rho_1^2}},\sqrt{\frac{2{\gamma_{th}}}{1-\rho_1^2}} \right) \right)^{M_R-1}}\\
&\times {\left( 1 - {Q_1}\left( \sqrt{\frac{2\rho_2^2{\gamma_{th}}}{1-\rho_2^2}},\sqrt{\frac{2{\gamma_{th}}}{1-\rho_2^2}} \right) \right)^{M_T-1}}\\
&\times {\left( 1 - {Q_1}\left( \sqrt{\frac{2\rho_1^2\rho_2^2{\gamma_{th}}}{1-\rho_1^2\rho_2^2}},\right.\right.}\\
&\qquad{\left.\left.\sqrt{\frac{2{\gamma_{th}}}{1-\rho_1^2\rho_2^2}} \right) \right)^{(M_R-1)(M_T-1)}} P_{out}^1.
\end{aligned}
\end{equation}
Finally, we set $M_R=M_T=1$ in the expression from \eqref{op2}. Substituting this result into \eqref{H2} then completes the proof.

\bibliographystyle{IEEEtran}
%\bibliography{IEEEabrv,FAS_RICIANBIB}

\begin{thebibliography}{10}
\providecommand{\url}[1]{#1}
\csname url@samestyle\endcsname
\providecommand{\newblock}{\relax}
\providecommand{\bibinfo}[2]{#2}
\providecommand{\BIBentrySTDinterwordspacing}{\spaceskip=0pt\relax}
\providecommand{\BIBentryALTinterwordstretchfactor}{4}
\providecommand{\BIBentryALTinterwordspacing}{\spaceskip=\fontdimen2\font plus
\BIBentryALTinterwordstretchfactor\fontdimen3\font minus
  \fontdimen4\font\relax}
\providecommand{\BIBforeignlanguage}[2]{{%
\expandafter\ifx\csname l@#1\endcsname\relax
\typeout{** WARNING: IEEEtran.bst: No hyphenation pattern has been}%
\typeout{** loaded for the language `#1'. Using the pattern for}%
\typeout{** the default language instead.}%
\else
\language=\csname l@#1\endcsname
\fi
#2}}
\providecommand{\BIBdecl}{\relax}
\BIBdecl

\bibitem{MIMO0}
G.~J. Foschini and M.~J. Gans, ``On limits of wireless communications in a fading environment when using multiple antennas,'' \emph{Wireless Pers. Commun.}, vol.~6, no.~3, pp. 311--335, Mar. 1998.

\bibitem{MIMO1}
A.~Paulraj, D.~Gore, R.~Nabar, and H.~Bölcskei, ``An overview of {MIMO} communications -- A key to gigabit wireless,'' \emph{Proc. IEEE}, vol.~92, no.~2, pp. 198--218, Feb. 2004.

\bibitem{MIMO2}
L.~Zheng and D.~Tse, ``Diversity and multiplexing: A fundamental tradeoff in multiple-antenna channels,'' \emph{IEEE Trans. Inf. Theory}, vol.~49, no.~5, pp. 1073--1096, May 2003.

\bibitem{Marzetta2010}
T.~L. Marzetta, ``Noncooperative cellular wireless with unlimited numbers of base station antennas,'' \emph{IEEE Trans. Wireless Commun.}, vol.~9, no.~11, pp. 3590--3600, Nov. 2010.

\bibitem{Massive1}
E.~G. Larsson, O.~Edfors, F.~Tufvesson, and T.~L. Marzetta, ``Massive {MIMO} for next generation wireless systems,'' \emph{IEEE Commun. Mag.}, vol.~52, no.~2, pp. 186--195, Feb. 2014.

\bibitem{Massive2}
H.~Q. Ngo, E.~G. Larsson, and T.~L. Marzetta, ``Energy and spectral efficiency of very large multiuser {MIMO} systems,'' \emph{IEEE Trans. Commun.}, vol.~61, no.~4, pp. 1436--1449, Apr. 2013.

\bibitem{Massive3}
L.~Lu, G.~Y. Li, A.~L. Swindlehurst, A.~Ashikhmin, and R.~Zhang, ``An overview of massive {MIMO}: Benefits and challenges,'' \emph{IEEE J. Sel. Topics Signal Process.}, vol.~8, no.~5, pp. 742--758, Oct. 2014.

\bibitem{6G1}
F.~Tariq \emph{et al.}, ``A speculative study on {6G},'' \emph{IEEE Wireless Commun.}, vol.~27, no.~4, pp. 118--125, Aug. 2020.

\bibitem{6G2}
Z.~Zhang \emph{et al.}, ``{6G} wireless networks: Vision, requirements, architecture, and key technologies,'' \emph{IEEE Veh. Technol. Mag.}, vol.~14, no.~3, pp. 28--41, Sept. 2019.

\bibitem{6G3}
W.~Saad, M.~Bennis, and M.~Chen, ``A vision of {6G} wireless systems: Applications, trends, technologies, and open research problems,'' \emph{IEEE Netw.}, vol.~34, no.~3, pp. 134--142, May/Jun. 2020.

\bibitem{Wang-xlmimo}
Z.~Wang \emph{et al.}, ``Extremely large-scale {MIMO}: Fundamentals, challenges, solutions, and future directions,'' \emph{IEEE Wireless Commun.}, vol.~31, no.~3, pp. 117--124, Jun. 2024.

\bibitem{Wang2024}
Z.~Wang \emph{et al.}, ``A tutorial on extremely large-scale {MIMO} for {6G}: Fundamentals, signal processing, and applications,'' \emph{IEEE Commun. Surveys Tuts.}, vol.~26, no.~3, pp. 1560--1605, Thirdquarter 2024.

\bibitem{NOMA}
Z.~Ding \emph{et al.}, ``Application of non-orthogonal multiple access in {LTE} and {5G} networks,'' \emph{IEEE Commun. Mag.}, vol.~55, no.~2, pp. 185--191, Feb. 2017.

\bibitem{RIS}
J.~Zhao, ``A survey of intelligent reflecting surfaces ({IRS}s): Towards {6G} wireless communication networks,'' Nov. 2019, \emph{arXiv:1907.04789}. [Online]. Available: \url{https://arxiv.org/abs/1907.04789}

\bibitem{MIMO_RIS_NOMA}
J.~Li \emph{et al.}, ``An {RIS}-aided interference mitigation-based design for {MIMO}-{NOMA} in cellular networks,'' \emph{IEEE Trans. Green Commun. Netw.}, vol.~8, no.~1, pp. 317--329, Mar. 2024.

\bibitem{MIMO_LIS}
T.~Hou \emph{et al.}, ``Performance analysis for large intelligent surfaces enabled {MIMO} networks,'' in \emph{Proc. IEEE Int. Conf. Commun. (ICC)}, Dublin, Ireland, Jun. 2020, pp. 1--6.

\bibitem{MIMO_RIS}
T.~Hou \emph{et al.}, ``{MIMO} assisted networks relying on intelligent reflective surfaces: A stochastic geometry based analysis,'' \emph{IEEE Trans. Veh. Technol.}, vol.~71, no.~1, pp. 571--582, Jan. 2022.

\bibitem{Sharawi2013}
M.~S. Sharawi, ``Printed multi-band {MIMO} antenna systems and their performance metrics,'' \emph{IEEE Antennas Propag. Mag.}, vol.~55, no.~5, pp. 218--232, Oct. 2013.

\bibitem{couple}
P.~S. Taluja and B.~L. Hughes, ``Diversity limits of compact broadband multi-antenna systems,'' \emph{IEEE J. Sel. Areas Commun.}, vol.~31, no.~2, pp. 326--337, Feb. 2013.

\bibitem{Alkhateeb2014}
A.~Alkhateeb, J.~Mo, N.~Gonzalez-Prelcic, and R.~W. Heath, ``{MIMO} precoding and combining solutions for millimeter-wave systems,'' \emph{IEEE Commun. Mag.}, vol.~52, no.~12, pp. 122--131, Dec. 2014.

\bibitem{Han2015}
S.~Han, I.~Chih-Lin, Z.~Xu, and C.~Rowell, ``Large-scale antenna systems with hybrid analog and digital beamforming for millimeter wave {5G},'' \emph{IEEE Commun. Mag.}, vol.~53, no.~1, pp. 186--194, Jan. 2015.

\bibitem{expensive}
R.~Zi, X.~Ge, J.~Thompson, C.-X. Wang, H.~Wang, and T.~Han, ``Energy efficiency optimization of {5G} radio frequency chain systems,'' \emph{IEEE J. Sel. Areas Commun.}, vol.~34, no.~4, pp. 758--771, Apr. 2016.

\bibitem{FAS}
K.-K. Wong, A.~Shojaeifard, K.-F. Tong, and Y.~Zhang, ``Fluid antenna systems,'' \emph{IEEE Trans. Wireless Commun.}, vol.~20, no.~3, pp. 1950--1962, Mar. 2021.

\bibitem{FAS1}
K.-K. Wong, K.-F. Tong, Y.~Zhang, and Z.~Zheng, ``Fluid antenna system for {6G}: When Bruce Lee inspires wireless communications,'' \emph{Electron. Lett.}, vol.~56, no.~24, pp. 1288--1290, Nov. 2020.

\bibitem{FAS2}
K.-K. Wong, K.-F. Tong, Y.~Shen, Y.~Chen, and Y.~Zhang, ``Bruce {Lee}-inspired fluid antenna system: Six research topics and the potentials for {6G},'' \emph{Front. Commun. Netw.}, vol.~3, Art. no. 853416, Mar. 2022.

\bibitem{block}
P.~Ramírez-Espinosa, D.~Morales-Jimenez, and K.-K. Wong, ``A new spatial block-correlation model for fluid antenna systems,'' \emph{IEEE Trans. Wireless Commun.}, early access, 2024, doi: 10.1109/TWC.2024.3434509.

\bibitem{mecha1}
L.~Zhu, W.~Ma, and R.~Zhang, ``Movable antennas for wireless communication: Opportunities and challenges,'' \emph{IEEE Commun. Mag.}, vol.~62, no.~6, pp. 114--120, Jun. 2024.

\bibitem{mecha2}
A.~Zhuravlev, V.~Razevig, S.~Ivashov, A.~Bugaev, and M.~Chizh, ``Experimental simulation of multi-static radar with a pair of separated movable antennas,'' in \emph{Proc. IEEE Int. Conf. Microw., Commun., Antennas Electron. Syst. (COMCAS)}, Tel Aviv, Israel, Nov. 2015, pp. 1--5.

\bibitem{mecha3}
X.~Li, Y.~Zhou, Z.~Shen, B.~Song, and S.~Li, ``Using a moving antenna to improve {GNSS/INS} integration performance under low-dynamic scenarios,'' \emph{IEEE Trans. Intell. Transp. Syst.}, vol.~23, no.~10, pp. 17717--17728, Oct. 2022.

\bibitem{liquid1}
G.~J. Hayes, J.-H. So, A.~Qusba, M.~D. Dickey, and G.~Lazzi, ``Flexible liquid metal alloy ({EGaIn}) microstrip patch antenna,'' \emph{IEEE Trans. Antennas Propag.}, vol.~60, no.~5, pp. 2151--2156, May 2012.

\bibitem{liquid2}
A.~M. Morishita, C.~K.~Y. Kitamura, A.~T. Ohta, and W.~A. Shiroma, ``A liquid-metal monopole array with tunable frequency, gain, and beam steering,'' \emph{IEEE Antennas Wireless Propag. Lett.}, vol.~12, pp. 1388--1391, 2013.

\bibitem{liquid3}
A.~Dey, R.~Guldiken, and G.~Mumcu, ``Microfluidically reconfigured wideband frequency-tunable liquid-metal monopole antenna,'' \emph{IEEE Trans. Antennas Propag.}, vol.~64, no.~6, pp. 2572--2576, Jun. 2016.

\bibitem{liquidFAS1}
Y.~Shen, K.-F. Tong, and K.-K. Wong, ``Reconfigurable surface wave fluid antenna for spatial {MIMO} applications,'' in \emph{Proc. IEEE-APS Topical Conf. Antennas Propag. Wireless Commun. (APWC)}, Honolulu, HI, USA, Aug. 2021, pp. 150--152.

\bibitem{liquidFAS2}
Y.~Shen, K.-F. Tong, and K.-K. Wong, ``Radiation pattern diversified single-fluid-channel surface-wave antenna for mobile communications,'' in \emph{Proc. IEEE-APS Topical Conf. Antennas Propag. Wireless Commun. (APWC)}, Cape Town, South Africa, Sept. 2022, pp. 49--51.

\bibitem{liquidFAS3}
Y.~Shen, K.-F. Tong, and K.-K. Wong, ``Radiation pattern diversified double-fluid-channel surface-wave antenna for mobile communications,'' in \emph{Proc. IEEE-APS Topical Conf. Antennas Propag. Wireless Commun. (APWC)}, Cape Town, South Africa, Sept. 2022, pp. 85--88.

\bibitem{EWOD}
H.~Wang, Y.~Shen, K.-F. Tong, and K.-K. Wong, ``Continuous electrowetting surface-wave fluid antenna for mobile communications,'' in \emph{Proc. IEEE Reg. 10 Conf. (TENCON)}, Hong Kong, China, Nov. 2022, pp. 1--3.

\bibitem{pixel1}
S.~Song and R.~D. Murch, ``An efficient approach for optimizing frequency reconfigurable pixel antennas using genetic algorithms,'' \emph{IEEE Trans. Antennas Propag.}, vol.~62, no.~2, pp. 609--620, Feb. 2014.

\bibitem{pixel2}
L.~Jing, M.~Li, and R.~Murch, ``Compact pattern reconfigurable pixel antenna with diagonal pixel connections,'' \emph{IEEE Trans. Antennas Propag.}, vol.~70, no.~10, pp. 8951--8961, Oct. 2022.

\bibitem{pixel3}
L.~N. Pringle \emph{et al.}, ``A reconfigurable aperture antenna based on switched links between electrically small metallic patches,'' \emph{IEEE Trans. Antennas Propag.}, vol.~52, no.~6, pp. 1434--1445, Jun. 2004.

\bibitem{tutorial}
W.~K. New \emph{et al.}, ``A tutorial on fluid antenna system for {6G} networks: Encompassing communication theory, optimization methods and hardware designs,'' \emph{IEEE Commun. Surveys Tuts.}, vol.~27, no.~4, pp. 2325--2377, Thirdquarter 2025.

\bibitem{FAS_ER}
K.-K. Wong, A.~Shojaeifard, K.-F. Tong, and Y.~Zhang, ``Performance limits of fluid antenna systems,'' \emph{IEEE Commun. Lett.}, vol.~24, no.~11, pp. 2469--2472, Nov. 2020.

\bibitem{FAS_partI}
K.-K. Wong, W.~K. New, X.~Hao, K.-F. Tong, and C.-B. Chae, ``Fluid antenna system---Part {I}: Preliminaries,'' \emph{IEEE Commun. Lett.}, vol.~27, no.~8, pp. 1919--1923, Aug. 2023.

\bibitem{FAS_partII}
K.-K. Wong, K.-F. Tong, and C.-B. Chae, ``Fluid antenna system---Part {II}: Research opportunities,'' \emph{IEEE Commun. Lett.}, vol.~27, no.~8, pp. 1924--1928, Aug. 2023.

\bibitem{wong2022closed}
K.-K. Wong, K.-F. Tong, Y.~Chen, and Y.~Zhang, ``Closed-form expressions for spatial correlation parameters for performance analysis of fluid antenna systems,'' \emph{Electron. Lett.}, vol.~58, no.~11, pp. 454--457, Apr. 2022.

\bibitem{Khammassi-2023}
M.~Khammassi, A.~Kammoun, and M.-S. Alouini, ``A new analytical approximation of the fluid antenna system channel,'' \emph{IEEE Trans. Wireless Commun.}, vol.~22, no.~12, pp. 8843--8858, Dec. 2023.

\bibitem{Zhao2025FAS}
H.~Zhao and D.~Slock, ``Analytical insights into outage probability and ergodic capacity of fluid antenna systems,'' \emph{IEEE Wireless Commun. Lett.}, vol.~14, no.~5, pp. 1581--1585, May 2025.

\bibitem{FAS_OP_diversity}
J.~D. Vega-S\'{a}nchez, A.~E. L\'{o}pez-Ram\'{i}rez, L.~Urquiza-Aguiar, and D.~P.~M. Osorio, ``Novel expressions for the outage probability and diversity gains in fluid antenna system,'' \emph{IEEE Wireless Commun. Lett.}, vol.~13, no.~2, pp. 372--376, Feb. 2024.

\bibitem{nakagami1}
L.~Tlebaldiyeva, G.~Nauryzbayev, S.~Arzykulov, A.~Eltawil, and T.~Tsiftsis, ``Enhancing {QoS} through fluid antenna systems over correlated Nakagami-$m$ fading channels,'' in \emph{Proc. IEEE Wireless Commun. Netw. Conf. (WCNC)}, Austin, TX, USA, Apr. 2022, pp. 78--83.

\bibitem{nakagami2}
J.~D. Vega-Sánchez, L.~Urquiza-Aguiar, M.~C.~P. Paredes, and D.~P.~M. Osorio, ``A simple method for the performance analysis of fluid antenna systems under correlated Nakagami-$m$ fading,'' \emph{IEEE Wireless Commun. Lett.}, vol.~13, no.~2, pp. 377--381, Feb. 2024.

\bibitem{copula}
F.~R. Ghadi, K.-K. Wong, F.~J. López-Martínez, C.-B. Chae, K.-F. Tong, and Y.~Zhang, ``A Gaussian copula approach to the performance analysis of fluid antenna systems,'' \emph{IEEE Trans. Wireless Commun.}, vol.~23, no.~11, pp. 17573--17585, Nov. 2024.

\bibitem{rician}
J.~Huangfu \emph{et al.}, ``Performance analysis of fluid antenna system under spatially-correlated Rician fading channels,'' \emph{IEEE Trans. Wireless Commun.}, early access, 2025, doi: 10.1109/TWC.2025.3590722.

\bibitem{FAS_alpha_mu}
P.~D. Alvim \emph{et al.}, ``On the performance of fluid antenna systems under $\alpha$-$\mu$ fading channels,'' \emph{IEEE Wireless Commun. Lett.}, vol.~13, no.~1, pp. 108--112, Jan. 2024.

\bibitem{FAS_ASER}
N.~Kapucu and M.~Bilim, ``{ASER} analysis of fluid antenna systems with rectangular and hexagonal {QAM} schemes,'' \emph{AEU--Int. J. Electron. Commun.}, vol.~196, Art. no. 155792, Jun. 2025.

\bibitem{MRC}
X.~Lai, T.~Wu, J.~Yao, C.~Pan, M.~Elkashlan, and K.-K. Wong, ``On performance of fluid antenna system using maximum ratio combining,'' \emph{IEEE Commun. Lett.}, vol.~28, no.~2, pp. 402--406, Feb. 2024.

\bibitem{trade}
W.~K. New, K.-K. Wong, H.~Xu, K.-F. Tong, and C.-B. Chae, ``An information-theoretic characterization of {MIMO-FAS}: Optimization, diversity-multiplexing tradeoff and $q$-outage capacity,'' \emph{IEEE Trans. Wireless Commun.}, vol.~23, no.~6, pp. 5541--5556, Jun. 2024.

\bibitem{marcumq}
M.~K. Simon, \emph{Probability Distributions Involving Gaussian Random Variables: A Handbook for Engineers and Scientists}. New York, NY, USA: Springer, 2002.

\bibitem{integral}
I.~S. Gradshteyn and I.~M. Ryzhik, \emph{Table of Integrals, Series, and Products}, 6th ed. New York, NY, USA: Academic Press, 2000.



\end{thebibliography}

% Generated by IEEEtran.bst, version: 1.14 (2015/08/26)

\end{document}